\documentclass[11pt]{amsart}
\usepackage{amsthm}
\usepackage{amsmath,amstext,amsthm,amsfonts,amssymb,amsthm}
\usepackage{mathrsfs}
\usepackage{multicol}
\usepackage{latexsym}
\usepackage{color}
\usepackage{graphicx}
\usepackage{epsf}
\usepackage{epsfig}
\usepackage{epic}
\usepackage{graphics}
\usepackage{verbatim}
\usepackage{color}
\usepackage{hyperref}
\newtheorem{theorem}{Theorem}[section]
\newtheorem{lemma}[theorem]{Lemma}
\newtheorem{corollary}[theorem]{Corollary}
\theoremstyle{definition}
\newtheorem{definition}[theorem]{Definition}
\newtheorem{proposition}[theorem]{Proposition}
\newtheorem{example}[theorem]{Example}

\theoremstyle{remark}
\newtheorem{remark}[theorem]{Remark}
\numberwithin{equation}{section}

\newcommand{\HI}{\mathfrak{H}}

\newcommand{\C}{\mathbb{C}}

\newcommand{\B}{\mathcal{B}}

\newcommand{\oz}{\overline{z}}
\newcommand{\qu}{\mathbf{q}}

\newcommand{\quat}{\mathbb H}

\newcommand{\mc}{\mathcal}

\newcommand{\be}{\begin{equation}}
\newcommand{\en}{\end{equation}}
\newcommand{\D}{{\mc D}}

\newcommand{\N}{\mathbb N}

\newcommand{\bedefin}{\begin{defi}}
	\newcommand{\findefi}{\end{defi} \medskip}
\newcommand{\betheo}{\begin{theorem}$\!\!${\bf \,\,\,}}
	\newcommand{\entheo}{\end{theorem}}
\newcommand{\enth}{\end{theorem}}
\newcommand{\becor}{\begin{cor}$\!\!${\bf .}}
	\newcommand{\encor}{\end{cor}}
\newcommand{\belem}{\begin{lem}$\!\!${\bf .}}
	\newcommand{\enlem}{\end{lem}}
\newcommand{\bea}{\begin{eqnarray}}
\newcommand{\ena}{\end{eqnarray}}
\newcommand{\beano}{\begin{eqnarray*}}
	\newcommand{\enano}{\end{eqnarray*}}
\newcommand{\bee}{\begin{enumerate}}
	\newcommand{\ene}{\end{enumerate}}
\newcommand{\bei}{\begin{itemize}}
	\newcommand{\eni}{\end{itemize}}
\newcommand{\betab}{\begin{tabular}}
	\newcommand{\entab}{\end{tabular}}
\newcommand{\bd}{\begin{displaymath}}

\newcommand{\Iop}{{\mathbb{I}_{V_{\mathbb{H}}^{R}}}}

\newcommand{\bsigma}{\mbox{\boldmath $\sigma$}}

\newcommand{\bfrakq}{\mbox{\boldmath $\mathfrak q$}}
\newcommand{\bfrakQ}{\mbox{\boldmath $\mathfrak Q$}}
\newcommand{\bfrakx}{\mbox{\boldmath $\mathfrak x$}}
\newcommand{\bfraky}{\mbox{\boldmath $\mathfrak y$}}
\newcommand{\bfraka}{\mbox{\boldmath $\mathfrak a$}}
\newcommand{\bfrakb}{\mbox{\boldmath $\mathfrak b$}}

\newcommand{\bfrake}{\mbox{\boldmath $\mathfrak e$}}
\newcommand{\bfrakp}{\mbox{\boldmath $\mathfrak p$}}

\newcommand{\bk}{\mathbf k}

\newcommand{\bq}{\mathbf q}

\newcommand{\bi}{\mathbf i}
\newcommand{\bj}{\mathbf j}
\newcommand{\di}{\text{diag}}
\begin{document}
\title[Deficiency Indices]{ Deficiency Indices of Some Classes of Unbounded $\quat$-Operators}
\author{B. Muraleetharan$^\dagger$, K. Thirulogasanthar$^\ddagger$}
\address{$^{\dagger}$ Department of Mathematics and Statistics, University of Jaffna, Thirunelveli, Jaffna, Sri Lanka. }
\address{$^{\ddagger}$ Department of Computer Science and Software Engineering, Concordia University, 1455 de Maisonneuve Blvd. West, Montreal, Quebec, H3G 1M8, Canada.}
\email{bbmuraleetharan@jfn.ac.lk, santhar@gmail.com}
%\thanks{This research is part of an M.Phil. thesis to be submitted to University of Jaffna }
\date{\today}
\begin{abstract}
In this paper we define the deficiency indices of a closed symmetric right $\quat$-linear operator and formulate a general theory of deficiency indices in a right quaternionic Hilbert space. This study provides a necessary and sufficient condition in terms of deficiency indices and in terms of S-spectrum, parallel to their complex counterparts, for a symmetric right $\quat$-linear operators to be self-adjoint.
\end{abstract}
\thanks{One of the auothors, KT would like to thank The Fonds de recherche du Québec - Nature et technologies (FRQNT) for partial financial support}
\maketitle
\pagestyle{myheadings}
%%%%%%%%%%%%%%%%%%%%%%%%%%%%%%%%%%%%%%%%%%%%%%%%%%%%%%%%%%%%%%%%%%%%%%%%
\section{Introduction}
Most of the operators that arise naturally in science are not bounded. They occur in numerous applications, remarkably in the theory of differential equations and in quantum mechanics.
Particularly closed operators and closable operators are most important classes of unbounded linear operators which are large enough to cover all interesting operators occurring in applications. However, the notion of being closed alone is too general. The concept of symmetric and self-adjoint operators play important role in many applications.
In particular, according to the Dirac-von Neumann formalism of complex quantum mechanics, the quantum mechanical observables such as position, momentum and spin are represented by self-adjoint unbounded operators on a complex Hilbert space \cite{Ali, Gaz,Schm}.  In analogy with complex quantum mechanics, states of quaternionic quantum mechanics are described by vectors of a separable quaternionic Hilbert space and observables in quaternionic quantum mechanics are represented by quaternion linear and self-adjoint operators \cite{Ad}.

The following question arises in several contexts: if an operator $A$ on a Hilbert space is symmetric, when does it have self-adjoint extensions? In the complex case, an answer is provided by the Cayley transform of a self-adjoint operator and the deficiency indices.
The theory of deficiency indices of closed symmetry operator in a complex Hilbert space is well-known and well-studied \cite{Schm}. Deficiency indices measure how far a symmetric operator is from being self-adjoint. Determining whether or not a symmetric operator is self-adjoint is important in physical applications because different self-adjoint extensions of the same operator yield different descriptions of the same system under consideration \cite{Schm}. The deficiency indices $(n_{+}(A),n_{-}(A))$ of a closed symmetric operator $A$ in a complex Hilbert space $\mathfrak{H}$ are defined by $$n_{\pm}(A)=\dim\ker(A^{\dagger}\mp i\mathbb{I}_{\mathfrak{H}});$$ where $A^{\dagger}$ is the adjoint of $A$ and $\mathbb{I}_{\mathfrak{H}}$ is the identity operator on $\mathfrak{H}$. Initially the distinction between symmetry and self-adjointness of an operator was poorly understood. However, von Neumann resolved this issue by introducing the deficiency spaces and deficiency indices. If the closure of A is self-adjoint, we say that A is essentially self-adjoint; and this corresponds to deficiency indices (0, 0). In fact, deficiency indices play a crucial role in the theory of self-adjoint extensions of symmetric operators in complex Hilbert spaces \cite{Schm}.

Birkhoff and von Neumann's theory of general quantum mechanics can be realised on Hilbert spaces over the fields $ \mathbb{R}$, the set of all real numbers, $\mathbb{C}$ , the set of all complex numbers, and $\quat$ , the set of all quaternions only \cite{Ad}. The fields $\mathbb{R}$ and $\mathbb{C}$ are associative and commutative and the theory of functional analysis is a well formed theory over real and complex Hilbert spaces. But the quaternions form a non-commutative associative algebra and this feature highly restricted mathematicians to work out a well-formed theory of functional analysis on quaternionic Hilbert spaces. Moreover, because of the non-commutativity of quaternions, quaternionic Hilbert spaces are formed by right or left multiplication of vectors by quaternionic scalars. In most cases, the two different conventions yield isomorphic versions of the theory \cite{Ad}.

To the best of our knowledge, a general theory of deficiency indices or self-adjoint extensions on quaternionic Hilbert spaces is not formulated yet. In this paper, we shall construct a general theory of deficiency indices for a closed symmetric right $\quat$-linear operator on a right quaternionic Hilbert space. We shall also characterize self-adjointness of an operator in terms of the so-called S-spectrum of the operator. In fact we shall provide necessary and sufficient condition, in terms of deficiency indices (Corollary \ref{se-ad-1} ) and in terms of the S-spectrum (Theorem \ref{S-spec}), for a closed symmetric $\quat$-linear operator to be self-adjoint.
%%%%%%%%%%%%%%%%%%%%%%%%%%%%%%%%%%%%%%%%%%%%%%%%%%%%%%%%%%%%%%%%%%%%%%
\section{Mathematical preliminaries}
In order to make the paper self-contained, we recall  few facts about quaternions which may not be well-known. In particular, we revisit the $2\times 2$ complex matrix representations of quaternions and quaternionic Hilbert spaces. For details we refer the reader to \cite{Ad,ghimorper,Vis,Zhang}.
\subsection{Quaternions}
Let $\quat$ denote the field of all quaternions and $\quat^*$ the group (under quaternionic multiplication) of all invertible quaternions. A general quaternion can be written as
$$\bfrakq = q_0 + q_1 \bi + q_2 \bj + q_3 \bk, \qquad q_0 , q_1, q_2, q_3 \in \mathbb R, $$
where $\bi,\bj,\bk$ are the three quaternionic imaginary units, satisfying
$\bi^2 = \bj^2 = \bk^2 = -1$ and $\bi\bj = \bk = -\bj\bi,  \; \bj\bk = \bi = -\bk\bj,
\; \bk\bi = \bj = - \bi\bk$. The quaternionic conjugate of $\bfrakq$ is
$$ \overline{\bfrakq} = q_0 - \bi q_1 - \bj q_2 - \bk q_3 . $$
We shall use the
$2\times 2$ matrix representation of the quaternions \cite{Zhang}, in which
$$
\bi = \sqrt{-1}\sigma_1, \quad \bj = -\sqrt{-1}\sigma_2, \quad \bk  = \sqrt{-1}\sigma_3, $$
and the $\sigma$'s are the three Pauli matrices,
$$
\sigma_1 = \begin{pmatrix} 0 & 1\\ 1& 0 \end{pmatrix}, \quad
\sigma_2 = \begin{pmatrix} 0 & -i\\ i & 0\end{pmatrix}, \quad
\sigma_3 = \begin{pmatrix} 1 & 0\\ 0 & -1 \end{pmatrix}, $$
to which we add
$$\sigma_0 = \mathbb I_2 = \begin{pmatrix} 1 & 0 \\ 0 & 1\end{pmatrix}.$$
We shall also use the matrix valued vector $\bsigma = (\sigma_1, -\sigma_2 , \sigma_3)$. Thus,
in this representation,
$$\bfrakq = q_0\sigma_0 + i\bq\cdot \bsigma =
\begin{pmatrix} q_0 + iq_3 & -q_2 +iq_1 \\ q_2 + iq_1 & q_0 -iq_3\end{pmatrix}, \quad
\bq = (q_1, q_2, q_3 ).  $$
In this case, the quaternionic conjugate of $\bfrakq$ is given by $\bfrakq^\dag$.
Introducing  two complex variables, which we write as
$$ z_1 = q_0 + iq_3 , \qquad z_2 = q_2 + iq_1, $$
we may also write
\be
\bfrakq = \begin{pmatrix} z_1 & -\overline{z}_2\\ z_2 & \overline{z}_1 \end{pmatrix}.
\label{quat-rep-comp}
\en
From (\ref{quat-rep-comp}) we get
\be
\text{det}[\bfrakq] = \vert z_1\vert^2 + \vert z_2\vert^2 = q_0^2 + q_1^2 + q_2^2 + q_3^2
: = \vert \bfrakq \vert^2 ,
\label{quat-norm}
\en
$\vert \bfrakq \vert$ denoting the usual norm of the quaternion $\bfrakq$. Note also that
$$
\bfrakq^\dag \bfrakq =   \bfrakq \bfrakq^\dag = \vert \bfrakq \vert^2\; \mathbb I_2 . $$
If $\bfrakq$ is non null element, it has the inverse
$$
\bfrakq^{-1} =  \frac 1{\vert \bfrakq \vert^2 }
\begin{pmatrix} \overline{z}_1  &  \overline{z}_2 \\ -z_2 & z_1 \end{pmatrix} .$$
On the other hand, we define $\bfrakQ:\C^{2}\longrightarrow\mathbb{M}_{2}(\C)$ by
\be
\bfrakQ(z)=\bfrakQ(z_{1},z_{2})= \begin{pmatrix} z_1 & -\overline{z}_2\\ z_2 & \overline{z}_1 \end{pmatrix},
\label{comp-rep-quat}
\en
for all $z=(z_{1},z_{2})\in\C^{2}$. Then the set ran$(\bfrakQ)=\{\bfrakQ(z)~\mid~z\in\C^{2}\}$ forms the algebra $\quat$ of quaternions because of the decomposition $$\bfrakQ(z)=\text{Re}z_{1}+\text{Im}z_{2}\,\bi+\text{Re}z_{2}\,\bj+\text{Im}z_{1}\,\bk.$$ This formalism allows us to look at the elements of $\quat$ as matrices and to use some operations in $\mathbb{M}_{2}(\C)$ rather than in $\quat$ (we refer the reader to \cite{Vas} for more details).
%%%%%%%%%%%%%%%%%%%%%%%%%%%%%%%%%%%%%%%%%%%%%%%%%%%%%%%%%
\subsection{$\quat$- Hilbert spaces}
In this subsection we  define left and right quaternionic Hilbert spaces and the Hilbert space of square-integrable functions on quaternions. For details we refer the reader to \cite{Ad,ghimorper,Vis}.
\subsubsection{Right $\quat$- Hilbert Space}
Let $V_{\quat}^{R}$ be a vector space under right multiplication by quaternions.  For $\phi,\psi,\omega\in V_{\quat}^{R}$ and $\bfrakq\in \quat$, the inner product
$$\langle\cdot\mid\cdot\rangle:V_{\quat}^{R}\times V_{\quat}^{R}\longrightarrow \quat$$
satisfies the following properties
\begin{enumerate}
	\item[(i)]
	$\overline{\langle \phi\mid \psi\rangle}=\langle \psi\mid \phi\rangle$
	\item[(ii)]
	$\|\phi\|^{2}=\langle \phi\mid \phi\rangle>0$ unless $\phi=0$, a real norm
	\item[(iii)]
	$\langle \phi\mid \psi+\omega\rangle=\langle \phi\mid \psi\rangle+\langle \phi\mid \omega\rangle$
	\item[(iv)]
	$\langle \phi\mid \psi\bfrakq\rangle=\langle \phi\mid \psi\rangle\bfrakq$
	\item[(v)]
	$\langle \phi\bfrakq\mid \psi\rangle=\overline{\bfrakq}\langle \phi\mid \psi\rangle$
\end{enumerate}
where $\overline{\bfrakq}$ stands for the quaternionic conjugate. It is always assumed that the
space $V_{\quat}^{R}$ is complete under the norm given above and separable. Then,  together with $\langle\cdot\mid\cdot\rangle$ this defines a right quaternionic Hilbert space (quaternionic linear spaces are sometimes called quaternionic modules too). Quaternionic Hilbert spaces share many of the standard properties of complex Hilbert spaces.

The next two Propositions can be established following the proof of their complex counterparts, refer to the reader \cite{ghimorper}.
\begin{proposition}\label{P1}
Let $\mathcal{O}=\{\varphi_{k}\,\mid\,k\in N\}$
be an orthonormal subset of $V_{\quat}^{R}$, where $N$ is a countable index set. Then following conditions are pairwise equivalent:
\begin{itemize}
\item [(a)] $\overline{\text{rightspan}\, \mathcal{O}}=V_{\quat}^{R}$.
\item [(b)] For every $\phi,\psi\in V_{\quat}^{R}$, then the series $\sum_{k\in N}\langle\phi\mid\varphi_{k}\rangle\langle\varphi_{k}\mid\psi\rangle$ converges absolutely and it holds:
$$\langle\phi\mid\psi\rangle=\sum_{k\in N}\langle\phi\mid\varphi_{k}\rangle\langle\varphi_{k}\mid\psi\rangle.$$
\item [(c)] For every  $\phi\in V_{\quat}^{R}$, it holds:
$$\|\phi\|^{2}=\sum_{k\in N}\mid\langle\varphi_{k}\mid\phi\rangle\mid^{2}.$$
\item [(d)] $\mathcal{O}^{\bot}=\{0\}$.
\end{itemize}
\end{proposition}
\begin{proposition}\label{P2}
(Every quaternionic Hilbert space) $V_{\quat}^{R}$ admits (a subset) $\mathcal{O}$, called Hilbert basis of $V_{\quat}^{R}$ (as in the preceding proposition), and $\mathcal{O}$ satisfies equivalent conditions (a)-(d) stated in the preceding proposition. Two such sets have the same cardinality.

Furthermore, if $\mathcal{O}$ is a Hilbert basis of $V_{\quat}^{R}$, then every  $\phi\in V_{\quat}^{R}$ can be uniquely decomposed as follows:
$$\phi=\sum_{k\in N}\varphi_{k}\langle\varphi_{k}\mid\phi\rangle,$$
where the series $\sum_{k\in N}\varphi_k\langle\varphi_{k}\mid\phi\rangle$ converges absolutely in $V_{\quat}^{R}$.
\end{proposition}

It should be noted that once a Hilbert basis is fixed, every left (resp. right) quaternionic Hilbert space also becomes a right (resp. left) quaternionic Hilbert space \cite{ghimorper,Vis}.

The field of quaternions $\quat$ itself can be turned into a left quaternionic Hilbert space by defining the inner product $\langle \bfrakq \mid \bfrakq^\prime \rangle = \bfrakq \bfrakq^{\prime\dag} = \bfrakq\overline{\bfrakq^\prime}$ or into a right quaternionic Hilbert space with  $\langle \qu \mid \qu^\prime \rangle = \bfrakq^\dag \bfrakq^\prime = \overline{\bfrakq}\bfrakq^\prime$.
%%%%%%%%%%%%%%%%%%%%%%%%%%%%%%%%%%%%%%%%%%%%%%%%%%%%%%%%%%%%%%%%
\subsubsection{$\quat$- Hilbert Spaces of Square-integrable Functions}
Let $(X, \mu)$ be a measure space and $\quat$  the field of quaternions, then
$$L^2_{\quat}(X,\mu)=\left\{f:X\rightarrow \quat\;\; \left| \;\; \int_X|f(x)|^2d\mu(x)<\infty \right.\right\}$$\label{L^2}
is a right quaternionic Hilbert space, with the (right) scalar product
\begin{equation}
\langle f \mid g\rangle =\int_X \overline{f(x)} g(x)\; d\mu(x),
\label{left-sc-prod}
\end{equation}
where $\overline{f(x)}$ is the quaternionic conjugate of $f(x)$, and (right)  scalar multiplication $f\bfraka , \; \bfraka\in \quat,$ with $(f \bfraka)(x) = f(x)\bfraka $ (see \cite{Vis} for details). Similarly, one could define a left quaternionic Hilbert space of square-integrable functions.
%%%%%%%%%%%%%%%%%%%%%%%%%%%%%%%%%%%%%%%%%%%%%%%%%%%%%%%%%%%%%%%%%%%%%%%%%%%%%%%%%%
\section{Right $\quat$-Operators and some basic properties}
In this section we shall define right  $\quat$-linear operators and acquire some basis properties from \cite{AC} and \cite{ghimorper} as needed for this manuscript. We shall also prove certain results pertinent to the development of the manuscript. To the best of our knowledge the results we prove do not appear in the literature.
\begin{definition}\cite{ghimorper}
A mapping $A:\D(A)\longrightarrow V_{\quat}^R$, where $\D(A)\subseteq V_{\quat}^R$ stands for the domain of $A$, is said to be right $\quat$-linear operator, if
$$A(\phi\bfraka+\psi\bfrakb)=(A\phi)\bfraka+(A\psi)\bfrakb,~~\mbox{~if~}~~\phi,\,\psi\in \D(A)~~\mbox{~and~}~~\bfraka,\bfrakb\in\quat.$$
\end{definition}
The set of all right linear operators will be denoted by $\mathcal{L}(V_{\quat}^{R})$ and the identity linear operator on $V_{\quat}^{R}$ will be denoted by $\Iop$. For a given $A\in \mathcal{L}(V_{\quat}^{R})$, the range and the kernel will be
\begin{eqnarray*}
\text{ran}(A)&=&\{\psi \in V_{\quat}^{R}~|~A\phi =\psi \quad\text{for}~~\phi \in\D(A)\}\\
\ker(A)&=&\{\phi \in\D(A)~|~A\phi =0\}.
\end{eqnarray*}
We call an operator $A\in \mathcal{L}(V_{\quat}^{R})$ bounded if
\begin{equation*}
\|A\|=\sup\{\|A\phi \|~|~\phi \in V_{\quat}^{R} \rm{~with~} \|\phi \|=1\}<\infty.
\end{equation*}
or equivalently, there exist $K\geq 0$ such that $\|A\phi \|\leq K\|\phi \|$ for $\phi \in\D(A)$. The set of all bounded right linear operators will be denoted by $\B(V_{\quat}^{R})$.
If $\D(A)$ is dense in $V_{\quat}^{R}$, then the adjoint $A^{\dagger}$ of $A$ is defined as
\begin{equation}\label{Ad1}
\langle \psi \mid A\phi \rangle=\langle A^{\dagger} \psi \mid\phi \rangle;\quad\text{for all}~~~(\phi ,\psi) \in \D(A)\times\D(A^{\dagger});
\end{equation}
where $\D(A^{\dagger})=\{\phi\in V_{\quat}^{R}~|~ \exists\,\omega_{\phi}\in V_{\quat}^{R} \rm{~such ~that~} \langle\omega_{\phi}\mid\psi\rangle=\langle\phi\mid A\psi\rangle, \forall\,\psi\in\D(A)\}$.
\subsection{Symmetry and Self-adjointness}
Let $A:\D(A)\longrightarrow V_{\quat}^R$ and $B:\D(B)\longrightarrow V_{\quat}^R$ be $\quat$-linear operators. As usual, we write $A\subset B$ if $D(A)\subset\D(B)$ and $B\vert_{\D(A)}=A$. In this case, $B$ is said to be an extension of $A$.
\begin{definition}\cite{ghimorper}
A right $\quat$-linear operator $A:\D(A)\longrightarrow V_{\quat}^R$ is said to be
\begin{itemize}
\item[(a)] \textit{symmetric}, if $A\subset A^{\dagger}$.
\item[(b)] \textit{anti-symmetric}, if $A\subset -A^{\dagger}$.
\item[(c)] \textit{self-adjoint}, if $A= A^{\dagger}$.
\item[(d)] \textit{unitary}, if $\D(A)=V_{\quat}^{R}$ and $A\,A^{\dagger}=A^{\dagger}A=\Iop$.
\item[(e)] \textit{closed}, if the graph $\mathcal{G}:=\D(A)\times\text{ran}(A)$ of $A$ is closed in $V_{\quat}^R\times V_{\quat}^R$, equipped with the product topology.
\item[(f)] \textit{closable}, if it admits closed operator extensions. In this case, the \textit{closure} $\overline{A}$ of $A$ is the smallest closed extension.
\end{itemize}
\end{definition}
\begin{proposition}\label{copcr}
A right $\quat$-linear operator $A:\D(A)\longrightarrow V_{\quat}^R$ is closed if and only if for any sequence $\{\phi_{n}\}$ in $\D(A)$ such that $\phi_{n}\longrightarrow\phi$ with $A\phi_{n}=\psi_{n}\longrightarrow\psi$ in $V_{\quat}^{R}$, then $\psi=A\phi$.
\end{proposition}
\begin{proof}
It is straight forward from the definition of closed operators.
\end{proof}
\begin{proposition}\cite{ghimorper}\label{cldop}
Let $A:\D(A)\longrightarrow V_{\quat}^R$ be densely defined right $\quat$-linear operator. Then
\begin{itemize}
\item[(a)] $A^{\dagger}$ is closed.
\item[(b)] $A$ is closable if and only if $\D(A^{\dagger})$ is dense in $V_{\quat}^R$, and $\overline{A}=A^{\dagger\dagger}$.
\item[(c)] $\text{ran}(A)^{\bot}=\ker(A^{\dagger})$ and $\ker(A)\subset\text{ran}(A^{\dagger})^{\bot}$.\\
Furthermore, if $\D(A^{\dagger})$ is dense in $V_{\quat}^R$ and $A$ is closed, then $\ker(A)=\text{ran}(A^{\dagger})^{\bot}$.
\end{itemize}
\end{proposition}
\begin{proposition}\label{prrpt}
Let $A:\D(A)\subseteq V_{\quat}^R\longrightarrow V_{\quat}^R$ be a right $\quat$-linear operator. If $A$ is closed and satisfies the condition that there exists $C>0$ such that
\be
\|A\phi\|\geq C\|\phi\|,
\label{rpt}
\en
for all $\phi\in D(A)$, then $\text{ran}(A)$ is closed.
\end{proposition}
\begin{proof}
Let $\psi\in\overline{\text{ran}(A)}$, then there exists $\{\phi_{n}\}$  a sequence in $\D(A)$ such that $A\phi_{n}\longrightarrow \psi$. Then by (\ref{rpt}), we know
$\{\phi_{n}\}$ is Cauchy sequence in $V_{\quat}^{R}$ as $\{A\phi_{n}\}$ is Cauchy. Then $\phi_{n}\longrightarrow \phi$ for some $\phi\in V_{\quat}^{R}$. From the Proposition \ref{copcr}, we have $A\phi=\psi$. This completes the proof.
\end{proof}
\subsection{Left Scalar Multiplications on $V_{\quat}^{R}$.}
We shall extract the definition and some properties of left scalar multiples of vectors on $V_{\quat}^R$ from \cite{ghimorper} as needed for the development of the manuscript. The left scalar multiple of vectors on a right quaternionic Hilbert space is an extremely non-canonical operation associated with a choice of preferred Hilbert basis. From the Proposition \ref{P2}, $V_{\quat}^{R}$ has a Hilbert basis
\begin{equation}\label{b1}
\mathcal{O}=\{\varphi_{k}\,\mid\,k\in N\},
\end{equation}
where $N$ is a countable index set.
The left scalar multiplication on $V_{\quat}^{R}$ induced by $\mathcal{O}$ is defined as the map $\quat\times V_{\quat}^{R}\ni(\bfrakq,\phi)\longmapsto \bfrakq\phi\in V_{\quat}^{R}$ given by
\begin{equation}\label{LPro}
\bfrakq\phi:=\sum_{k\in N}\varphi_{k}\bfrakq\langle \varphi_{k}\mid \phi\rangle,
\end{equation}
for all $(\bfrakq,\phi)\in\quat\times V_{\quat}^{R}$. Since all left multiplications are made with respect to some basis, assume that the basis $\mathcal{O}$ given by (\ref{b1}) is fixed all over the paper.
\begin{proposition}\cite{ghimorper}\label{lft_mul}
The left product defined in (\ref{LPro}) satisfies the following properties. For every $\phi,\psi\in V_{\quat}^{R}$ and $\bfrakp,\bfrakq\in\quat$,
\begin{itemize}
\item[(a)] $\bfrakq(\phi+\psi)=\bfrakq\phi+\bfrakq\psi$ and $\bfrakq(\phi\bfrakp)=(\bfrakq\phi)\bfrakp$.
\item[(b)] $\|\bfrakq\phi\|=|\bfrakq|\|\phi\|$.
\item[(c)] $\bfrakq(\bfrakp\phi)=(\bfrakq\bfrakp\phi)$.
\item[(d)] $\langle\overline{\bfrakq}\phi\mid\psi\rangle
=\langle\phi\mid\bfrakq\psi\rangle$.
\item[(e)] $r\phi=\phi r$, for all $r\in \mathbb{R}$.
\item[(f)] $\bfrakq\varphi_{k}=\varphi_{k}\bfrakq$, for all $k\in N$.
\end{itemize}
\end{proposition}
 \begin{remark}
 (1) The meaning of writing $\bfrakp\phi$ is $\bfrakp\cdot\phi$, because the notation from (\ref{LPro}) may be confusing, when $V_{\quat}^{R}=\quat$. However, regarding the field $\quat$ itself as a right $\quat$-Hilbert space, an orthonormal basis $\mathcal{O}$ should consist only of a singleton, say $\{\varphi_{0}\}$, with $\mid\varphi_{0}\mid=1$, because we clearly have $\theta=\varphi_{0}\langle\varphi_{0}\mid\theta\rangle$, for all $\theta\in\quat$. The equality from (f) of Proposition \ref{lft_mul} can be written as $\bfrakp\varphi_{0}=\varphi_{0}\bfrakp$, for all $\bfrakp\in\quat$. In fact, the left hand side may be confusing and it should be understood as $\bfrakp\cdot\varphi_{0}$, because the true equality $\bfrakp\varphi_{0}=\varphi_{0}\bfrakp$ would imply that $\varphi_{0}=\pm 1$. For simplicity in notation, we are writing $\bfrakp\phi$ instead of writing $\bfrakp\cdot\phi$.\\
 (2) Also one can trivially see that $(\bfrakp+\bfrakq)\phi=\bfrakp\phi+\bfrakq\phi$, for all $\bfrakp,\bfrakq\in\quat$ and $\phi\in V_{\quat}^{R}$.
 \end{remark}
Furthermore, the quaternionic scalar multiplication of $\quat$-linear operators is also defined in \cite{ghimorper}. For any fixed $\bfrakq\in\quat$ and a given right $\quat$-linear operator $A:\D(A)\longrightarrow V_{\quat}^R$, the left scalar multiplication of $A$ is defined as a map $\bfrakq A:\D(A)\longrightarrow V_{\quat}^R$ by the setting
\begin{equation}\label{lft_mul-op}
(\bfrakq A)\phi:=\bfrakq (A\phi)=\sum_{k\in N}\varphi_{k}\bfrakq\langle \varphi_{k}\mid A\phi\rangle,
\end{equation}
for all $\phi\in D(A)$. It is straightforward that $\bfrakq A$ is a right $\quat$-linear operator. If $\bfrakq\phi\in \D(A)$, for all $\phi\in \D(A)$, one can define right scalar multiplication of the right $\quat$-linear operator $A:\D(A)\longrightarrow V_{\quat}^R$ as a map $ A\bfrakq:\D(A)\longrightarrow V_{\quat}^R$ by the setting
\begin{equation}\label{rgt_mul-op}
(A\bfrakq )\phi:=A(\bfrakq \phi),
\end{equation}
for all $\phi\in D(A)$. It is also right $\quat$-linear operator. One can easily obtain that, if $\bfrakq\phi\in \D(A)$, for all $\phi\in \D(A)$ and $\D(A)$ is dense in $V_{\quat}^R$, then
\begin{equation}\label{sc_mul_aj-op}
(\bfrakq A)^{\dagger}=A^{\dagger}\overline{\bfrakq}~\mbox{~and~}~
(A\bfrakq)^{\dagger}=\overline{\bfrakq}A^{\dagger}.
\end{equation}
\begin{remark}
Using (\ref{lft_mul-op}) and (\ref{sc_mul_aj-op}), we have,
for any $\phi\in V_{\quat}^{R}$,
$$(\bi \Iop)^{\dagger}\phi=[(\Iop)^{\dagger}\overline{\bi}]\phi=(\Iop\overline{\bi})\phi=\Iop(\overline{\bi}\phi)=-\bi\phi=-\bi(\Iop\phi)=-(\bi \Iop)\phi.$$
That is, $(\bi \Iop)^{\dagger}=-(\bi \Iop)$. One can also see that $(\bj \Iop)^{\dagger}=-(\bj \Iop)$ and $(\bk \Iop)^{\dagger}=-(\bk \Iop)$ .
\end{remark}
\begin{lemma}\label{neq-i}
Let $A:\D(A)\subseteq V_{\quat}^R\longrightarrow V_{\quat}^R$ be a densely defined right $\quat$-linear operator with the property that $\bi\phi\in \D(A)$, for all $\phi\in \D(A)$. If the operator $\bi A$ is anti-symmetric, then
\be
\|(A\pm \bi\Iop)\phi\|^{2}=\|A\phi\|^{2}+\|\phi\|^{2},
\label{neq1}
\en
for all $\phi\in\D(A)$.
Moreover, if $A$ is symmetric and $\bi A$ is  anti-symmetric, then
\be
\bi A=A\bi.
\label{iA}
\en
\end{lemma}
\begin{proof}
Let $\phi\in\D(A)$, then
\begin{eqnarray*}
\|(A\pm \bi\Iop)\phi\|^{2}
&=&\langle(A\pm \bi\Iop)\phi\mid\langle(A\pm\bi\Iop)\phi\rangle\rangle\\
&=&\langle A\phi\mid A\phi \rangle\pm\langle A\phi\mid (\bi\Iop)\phi \rangle\pm\langle (\bi\Iop)\phi\mid A\phi \rangle+\langle (\bi\Iop)\phi\mid (\bi\Iop)\phi \rangle.
\end{eqnarray*}
For any $\phi\in\D(A)$, one can obtain that,
\begin{eqnarray*}
	\langle A\phi\mid (\bi\Iop)\phi \rangle
	&=&\langle A\phi\mid \bi\phi \rangle=\langle \overline{\bi}(A\phi)\mid \phi \rangle~\mbox{~by (d) in Proposition \ref{lft_mul}}\\
	&=&\langle -(\bi A)\phi\mid \phi \rangle=\langle (\bi A)^{\dagger}\phi\mid \phi \rangle~\mbox{~as $\bi A$ is anti-symmetric}\\
	&=&\langle (A^{\dagger}\overline{\bi})\phi\mid \phi \rangle=\langle A^{\dagger}(\overline{\bi}\phi)\mid \phi \rangle~\mbox{~by (\ref{sc_mul_aj-op})}\\
	&=&\langle \overline{\bi}\phi\mid A \phi \rangle=-\langle \bi\phi\mid A \phi \rangle=-\langle (\bi\Iop)\phi\mid A\phi \rangle.
\end{eqnarray*}
Therefore, $$\|(A\pm \bi\Iop)\phi\|^{2}=\langle A\phi\mid A\phi \rangle+\langle (\bi\Iop)\phi\mid (\bi\Iop)\phi \rangle.$$
This proves the equation (\ref{neq1}).
	Since $A$ is symmetric and $\bi A$ is anti-symmetric, we have, using (\ref{sc_mul_aj-op}), on $\D(A)$,
	$-\bi A=(\bi A)^{\dagger}=A^{\dagger}\overline{\bi}=A\overline{\bi}=-A\bi$, that is, $\bi A=A\bi.$
This completes the proof.
\end{proof}
The following proposition is a direct consequence of the above lemma.
\begin{proposition}\label{ken}
Let $A:\D(A)\subseteq V_{\quat}^R\longrightarrow V_{\quat}^R$ be a densely defined right $\quat$-linear operator with the property that $\bi\phi\in \D(A)$, for all $\phi\in \D(A)$. If the operator $\bi A$ is anti-symmetric, then  $$\ker{(A\pm \bi\Iop)}=\{0\}.$$
\end{proposition}
\begin{proof}
	Let $\phi\in\ker{(A\pm \bi\Iop)}$. Then by (\ref{neq1}), we have
	$$0=\|(A\pm \bi\Iop)\phi\|^{2}=\|A\phi\|^{2}+\|\phi\|^{2}.$$
	This implies that $\phi=0$.
\end{proof}

The following theorem provides a basic criterion for self-adjointness.
\begin{theorem}\label{csadj_i}
Let $A:\D(A)\subseteq V_{\quat}^R\longrightarrow V_{\quat}^R$ be a densely defined right $\quat$-linear operator with the property that $\bi\phi\in \D(A)$, for all $\phi\in \D(A)$. If the operator $A$ is symmetric and $\bi A$ is anti-symmetric, then the following statements are equivalent:
\begin{itemize}
\item[(a)] $A$ is self-adjoint.
\item[(b)] $A$ is closed and $\ker(A^{\dagger}\pm\bi\Iop)=\{0\}$.
\item[(c)] $\text{ran}(A\pm\bi\Iop)=V_{\quat}^{R}$.
\end{itemize}
\end{theorem}
\begin{proof}
(a) $\Rightarrow$ (b): One can easily see that $A^{\dagger}=A$ is closed from Proposition \ref{cldop}. By the Proposition \ref{ken}, we have $\ker{(A^{\dagger}\pm \bi\Iop)}=\ker{(A\pm \bi\Iop)}=\{0\}$.\\
(b) $\Rightarrow$ (c): Now $\ker{(A^{\dagger}\pm \bi\Iop)}=\{0\}$ implies with the Proposition \ref{cldop} that $\text{ran}(A\mp \bi\Iop)^{\bot}=\{0\}$. But from the equality (\ref{neq1}), we get
$$\|(A\pm \bi\Iop)\phi\|=\sqrt{\|A\phi\|^{2}+\|\phi\|^{2}}\geq\|\phi\|$$
and this together with Proposition \ref{prrpt} gives that $\text{ran}(A\pm\bi\Iop)$ is closed. Thereby $\text{ran}(A\pm\bi\Iop)=V_{\quat}^{R}$.\\
(c) $\Rightarrow$ (a): Using (c) with the Proposition \ref{cldop}, it is easy to see that $\ker(A^{\dagger}\pm\bi\Iop)=\{0\}$. Let $\phi\in\D(A^{\dagger})$, then $$(A^{\dagger}\pm\bi\Iop)(\phi)=(A\pm\bi\Iop)(\psi),$$ for some $\psi\in\D(A)$ as $\text{ran}(A\pm\bi\Iop)=V_{\quat}^{R}$. Since $A$ is symmetric, that is, $\D(A)\subseteq\D(A^{\dagger})$, we have
$$(A^{\dagger}\pm\bi\Iop)(\phi-\psi)=0.$$
This show that $\phi=\psi\in\D(A)$. Hence $\D(A)=\D(A^{\dagger})$ and therefore $A$ is self-adjoint.
\end{proof}
\subsection{Spectrum of Unbounded $\quat$-linear Operators} For a given right $\quat$-linear operator $A:\D(A)\subseteq V_{\quat}^R\longrightarrow V_{\quat}^R$ and $\bfrakq\in\quat$, we define the operator $R_{\bfrakq}(A):\D(A^{2})\longrightarrow\quat$ by  $$R_{\bfrakq}(A)=A^{2}-2\text{Re}(\bfrakq)A+|\bfrakq|^{2}\Iop,$$
where $\bfrakq=q_{0}+\bi q_1 + \bj q_2 + \bk q_3$ is a quaternion, $\text{Re}(\bfrakq)=q_{0}$  and $|\bfrakq|^{2}=q_{0}^{2}+q_{1}^{2}+q_{2}^{2}+q_{3}^{2}.$
\begin{definition}\cite{Fab, Fab1,ghimorper}
Let $A:\D(A)\subseteq V_{\quat}^R\longrightarrow V_{\quat}^R$ be a right $\quat$-linear operator. The \textit{spherical resolvent} set of $A$ is the set $\rho_{S}(A)\,(\subset\quat)$ such that the three following conditions hold true:
\begin{itemize}
\item[(a)] $\ker(R_{\bfrakq}(A))=\{0\}$.
\item[(b)] $\text{ran}(R_{\bfrakq}(A))$ is dense in $V_{\quat}^{R}$.
\item[(c)] $R_{\bfrakq}(A)^{-1}:\text{ran}(R_{\bfrakq}(A))\longrightarrow\D(A^{2})$ is bounded.
\end{itemize}
The \textit{spherical spectrum} $\sigma_{S}(A)$ of $A$ is defined by setting $\sigma_{S}(A):=\quat\smallsetminus\rho_{S}(A)$. It decomposes into three disjoint subsets as follows:
\begin{itemize}
\item[(i)] the \textit{spherical point spectrum} of $A$: $$\sigma_{pS}(A):=\{\bfrakq\in\quat~\mid~\ker(R_{\bfrakq}(A))\ne\{0\}\}.$$
\item[(ii)] the \textit{spherical residual spectrum} of $A$: $$\sigma_{rS}(A):=\{\bfrakq\in\quat~\mid~\ker(R_{\bfrakq}(A))=\{0\},\overline{\text{ran}(R_{\bfrakq}(A))}\ne V_{\quat}^{R}~\}.$$
\item[(ii)] the \textit{spherical residual spectrum} of $A$: $$\sigma_{cS}(A):=\{\bfrakq\in\quat~\mid~\ker(R_{\bfrakq}(A))=\{0\},\overline{\text{ran}(R_{\bfrakq}(A))}= V_{\quat}^{R}, R_{\bfrakq}(A)^{-1}\notin\B(V_{\quat}^{R}) ~\}.$$
\end{itemize}
If $A\phi=\phi\bfrakq$ for some $\bfrakq\in\quat$ and $\phi\in V_{\quat}^{R}\smallsetminus\{0\}$, then $\phi$ is called an \textit{eigenvector of $A$ with eigenvalue} $\bfrakq$.
\end{definition}
\begin{proposition}
\label{Pr2}\cite{AC, ghimorper}Let $A\in\mathcal{L}(V_H^R)$ and $A$ be self-adjoint, then $\sigma_S(A)\subset\mathbb{R}$.
\end{proposition}
%%%%%%%%%%%%%%%%%%%%%%%%%%%%%%%%%%%%%%%%%%%%%%%%%%%%%%%%%%%%%%%%%%%%
\section{Deficiency Indices}
\subsection{Von Neumann Theorem and Some Preliminary}
The following theorem is crucial to the development of the manuscript and it involves left multiple of vectors and operators. In order to enhance clarity we give a detailed proof.
\begin{theorem}\label{Von}(Von Neumann). Let $A:\D(A)\subseteq V_{\quat}^R\longrightarrow V_{\quat}^R$ be a densely defined closed right $\quat$-linear operator  with the property that $\bi\phi\in \D(A)$, for all $\phi\in \D(A)$. If the operator $A$ is symmetric and $\bi A$ is anti-symmetric, then
\be
\D(A^{\dagger})=\D(A)\oplus\ker(A^{\dagger}-\bi\Iop)\oplus\ker(A^{\dagger}+\bi\Iop).
\label{Von_eq}
\en
\end{theorem}
\begin{proof}
Let $\D_{\pm}:=\ker(A^{\dagger}\mp\bi\Iop)$. Now from the equality (\ref{neq1}), we get
$$\|(A\pm \bi\Iop)\phi\|=\sqrt{\|A\phi\|^{2}+\|\phi\|^{2}}\geq\|\phi\|$$
using this together with Proposition \ref{prrpt} we get $\text{ran}(A\pm\bi\Iop)$ is closed. From (c) in the Proposition \ref{cldop}, we have  $\text{ran}(A\pm\bi\Iop)^{\bot}=\ker(A^{\dagger}\mp\bi\Iop)$. To prove that $\D(A)$ and $\D_{\pm}$ are linearly independent subspaces, take $(\phi_{0},\phi_{\pm})\in\D(A)\times\D_{\mp}$ such that $\phi_{0}+\phi_{+}+\phi_{-}=0$. Then, as $A$ is symmetric, we have
\begin{eqnarray*}
(A-\bi\Iop)\phi_{0}&=&-[(A^{\dagger}-\bi\Iop)\phi_{-}+(A^{\dagger}-\bi\Iop)\phi_{+}]\\
&=&-[(A^{\dagger}-\bi\Iop)\phi_{-}+(A^{\dagger}+\bi\Iop)\phi_{+}-2(\bi\Iop)\phi_{+}]\\
&=&2\bi\phi_{+}\mbox{~~as~~}\phi_{\pm}\in\ker(A^{\dagger}\pm\bi\Iop).
\end{eqnarray*}
Now $\phi_{+}\in\ker(A^{\dagger}+\bi\Iop)\subseteq\D(A^{\dagger})\subseteq V_{\quat}^R=\overline{\D(A)}$. That is, there exists a sequence $\{\phi_{n}\}$ in $\D(A)$ such that $\phi_{n}\longrightarrow\phi_+$ as $n\longrightarrow\infty$. From this, for each $k\in N$, $\varphi_{k}\bi\langle \varphi_{k}\mid \phi_{n}\rangle\longrightarrow\varphi_{k}\bi\langle \varphi_{k}\mid \phi_{+}\rangle$ as $n\longrightarrow\infty$ and so
$$2\bi\phi_n =2\sum_{k\in N}\varphi_{k}\bi\langle \varphi_{k}\mid \phi_{n}\rangle\longrightarrow2\sum_{k\in N}\varphi_{k}\bi\langle \varphi_{k}\mid \phi_{+}\rangle=2\bi\phi_+ $$
as $n\longrightarrow\infty$. Now for each $n\in\mathbb{N}$, we have
\begin{eqnarray*}
(A^{\dagger}+\bi\Iop)(2\bi\phi_n)&=&A^{\dagger}(2\bi\phi_n)+\bi\Iop(2\bi\phi_n)\\
&=& 2[A^{\dagger}(\bi\phi_n)+\bi(\bi\phi_n)]\\
&=& 2[A^{\dagger}(\bi\phi_n)-\phi_n]\mbox{~~by (c) in Proposition \ref{lft_mul}}\\
&=& 2[(A^{\dagger}\bi)\phi_n-\phi_n]\mbox{~~by \ref{rgt_mul-op}}\\
&=& 2[(\overline{\bi} A)^{\dagger}\phi_n-\phi_n]\mbox{~~by \ref{sc_mul_aj-op}}\\
&=& 2[-(\bi A)^{\dagger}\phi_n-\phi_n]\\
&=& 2[(\bi A)\phi_n-\phi_n] \mbox{~~as $\bi A$ is anti-symmetric}\\
&=& 2[\bi (A\phi_n)+\bi((\bi\Iop)\phi_n)] \mbox{~~~by \ref{lft_mul-op}}\\
&=& 2\bi[(A\phi_n)+((\bi\Iop)\phi_n)]\mbox{~~by (a) in Proposition \ref{lft_mul}}\\
&=& 2\bi[(A+\bi\Iop)\phi_n].
\end{eqnarray*}
That is, since  $\{\phi_{n}\}\subseteq\D(A)$ and $A$ symmetric, we have for each $n\in\mathbb{N}$, $$(A^{\dagger}+\bi\Iop)(2\bi\phi_n)=2\bi[(A^{\dagger}+\bi\Iop)\phi_n].$$ Now let us show that $(A^{\dagger}+\bi\Iop)$ is a closed operator. For, let $\{\xi_{n}\}$ be a sequence in $\D(A^{\dagger})$ such that $\xi_{n}\longrightarrow\xi$ as $n\longrightarrow\infty$ with $(A^{\dagger}+\bi\Iop)\xi_{n}\longrightarrow\eta$ as $n\longrightarrow\infty$. Then  $(\bi\Iop)\xi_{n}=\bi\xi_{n}\longrightarrow\bi\xi=(\bi\Iop)\xi$ as $n\longrightarrow\infty$ and so
$$\|A^{\dagger}\xi_{n}-(\eta-\bi\xi)\|\leq\|(A^{\dagger}+\bi\Iop)\xi_{n}-\eta\|
+\|(\bi\Iop)\xi_{n}-(\bi\Iop)\xi\|\longrightarrow 0$$
 as $n\longrightarrow\infty$. That is,$A^{\dagger}\xi_{n}\longrightarrow(\eta-\bi\xi)$  as $n\longrightarrow\infty$. But we know $A^{\dagger}$ is closed. Thus $A^{\dagger}\xi=\eta-\bi\xi$ or equivalently $(A^{\dagger}+\bi\Iop)\xi=\eta$. Therefore, by Proposition \ref{copcr}, $(A^{\dagger}+\bi\Iop)$ is closed.

Assume that $(A^{\dagger}+\bi\Iop)\phi_n\longrightarrow\psi$ as $n\longrightarrow\infty$, for some $\psi\in V_{\quat}^R$. But $\phi_{n}\longrightarrow\phi_+$ as $n\longrightarrow\infty$. Thus $ \psi=(A^{\dagger}+\bi\Iop)\phi_+=0$ as $(A^{\dagger}+\bi\Iop)$ is a closed operator. That is, $(A^{\dagger}+\bi\Iop)\phi_n\longrightarrow 0$ as $n\longrightarrow\infty$. From this one can easily see that $2\bi[(A^{\dagger}+\bi\Iop)\phi_n]\longrightarrow 0$ as $n\longrightarrow\infty$. Therefore we have $(A^{\dagger}+\bi\Iop)(2\bi\phi_n)\longrightarrow 0$ as $n\longrightarrow\infty$ and $2\bi\phi_n \longrightarrow 2\bi\phi_+ $ as $n\longrightarrow\infty$. Thus $(A^{\dagger}+\bi\Iop)(2\bi\phi_+)=0$ as $(A^{\dagger}+\bi\Iop)$ is closed. That is,
$$2\bi\phi_+\in\ker(A^{\dagger}+\bi\Iop)=\text{ran}(A-\bi\Iop)^{\bot}.$$
Then $2\bi\phi_{+}\in\text{ran}(A-\bi\Iop)\cap\text{ran}(A-\bi\Iop)^{\bot}=\{0\}.$ That is, $$2\bi\phi_{+}=2\sum_{k\in N}\varphi_{k}\bi\langle \varphi_{k}\mid \phi_{+}\rangle=0.$$ Using the orthonormality of Hilbert basis, we can obtain that, for any $j\in N$,  $\langle \varphi_{j}\mid \phi_{+}\rangle=0$ and it suffices to say that $\phi_{+}=0$.
In a similar fashion we can prove $\phi_{-}=0$, and so $\phi_{0}=0$. One can trivially see that $$\D(A)\oplus\ker(A^{\dagger}-\bi\Iop)\oplus\ker(A^{\dagger}+\bi\Iop)\subseteq\D(A^{\dagger}).$$
On the other hand, take $\phi\in\D(A^{\dagger})$. Since $(A^{\dagger}-\bi\Iop)\phi\in V_{\quat}^{R}$ and $$V_{\quat}^{R}=\text{ran}(A-\bi\Iop)\oplus\text{ran}(A-\bi\Iop)^{\bot}=\text{ran}(A-\bi\Iop)\oplus\ker(A^{\dagger}+\bi\Iop)=\text{ran}(A-\bi\Iop)\oplus\D_{-},$$
we have there exists $(\theta,\omega)\in\text{ran}(A-\bi\Iop)\times\ker(A^{\dagger}+\bi\Iop)$ such that $(A^{\dagger}-\bi\Iop)\phi=\theta+\omega$. Now there exists $\phi_{0}\in\D(A)$ such that $\theta=(A-\bi\Iop)\phi_{0}$ and choose $\displaystyle\phi_+=\frac{1}{2}\bi\,\omega$. Then $\phi_{+}\in\D_- =\ker(A^{\dagger}+\bi\Iop)$ and $\displaystyle\bi\phi_+ =\bi\left(\frac{1}{2}\bi\,\omega\right)=-\frac{1}{2}\omega$. That is, $(\phi_{0},-2\bi\phi_{+})\in\D(A)\times\D_{-}$ such that
\begin{equation}\label{eqq1}
(A^{\dagger}-\bi\Iop)\phi=(A-\bi\Iop)\phi_{0}+(-2\bi\phi_{+}).
\end{equation}
Thus
\begin{eqnarray*}(A^{\dagger}-\bi\Iop)(\phi-\phi_{0}-\phi_{+})
&=& (A^{\dagger}-\bi\Iop)\phi-(A^{\dagger}-\bi\Iop)\phi_{0}-(A^{\dagger}-
\bi\Iop)\phi_+\\
&=& (A^{\dagger}-\bi\Iop)\phi-(A^{\dagger}-\bi\Iop)\phi_{0}-(A^{\dagger}+
\bi\Iop)\phi_+ +2\bi\Iop\phi_+\\
&=&(A^{\dagger}-\bi\Iop)\phi-(A-\bi\Iop)\phi_{0}+2\bi\phi_{+}\mbox{~~as~~}(A^{\dagger}+\bi\Iop)\phi_{+}=0\\
&=&0\mbox{~~by (\ref{eqq1})}.
\end{eqnarray*}
If we take $\phi_{-}=\phi-\phi_{0}-\phi_{+}\in\D_{+}$
 then the result follows.
\end{proof}
\begin{proposition}
	\label{neq1-i}
	Let $A:\D(A)\subseteq V_{\quat}^R\longrightarrow V_{\quat}^R$ be a densely defined right $\quat$-linear operator with the property that $\bi\phi\in \D(A)$, for all $\phi\in \D(A)$. If the operator $\bi A$ is anti-symmetric, then
	\be
	\|(A\pm \bi \lambda\Iop)\phi\|^{2}=\|A\phi\|^{2}+\mid \lambda\mid^{2}\|\phi\|^{2},
	\label{neq1-1}
	\en
	for all $\phi\in\D(A)$ and $\lambda\in\mathbb{R}$.
\end{proposition}\label{Gen-neq}
\begin{proof}
	If $\lambda\ne0$ and we replace the operator $A$ in (\ref{neq1}) by $\displaystyle\frac{1}{\lambda}\,A$, then
	$$\left\|\left(\left[ \frac{1}{\lambda}\,A\right] \pm \bi\Iop\right)\phi\right\|^{2}=\left\|\left[ \frac{1}{\lambda}\,A\right] \phi\right\|^{2}+\left\|\phi\right\|^{2}.$$
By (e) in Proposition \ref{lft_mul}, we have $\displaystyle\left[ \frac{1}{\lambda}\,A\right]\phi=\frac{1}{\lambda}\left[A\phi\right]=\left[A\phi\right]\frac{1}{\lambda}$. Thus it is easy to see that
$$\frac{1}{\mid\lambda\mid^{2}}\left\|\left(A \pm \bi\lambda\Iop\right)\phi\right\|^{2}=\frac{1}{\mid\lambda\mid^{2}}\left\|A \phi\right\|^{2}+\left\|\phi\right\|^{2}.$$ This concludes the results.
\end{proof}
\begin{theorem}\label{G-Von}
	Let $A:\D(A)\subseteq V_{\quat}^R\longrightarrow V_{\quat}^R$ be a densely defined closed right $\quat$-linear operator  with the property that $\bi\phi\in \D(A)$, for all $\phi\in \D(A)$. If the operator $A$ is symmetric and $\bi A$ is anti-symmetric, then
	\be
	\D(A^{\dagger})=\D(A)\oplus\ker(A^{\dagger}-\bi\lambda\Iop)
	\oplus\ker(A^{\dagger}+\bi\lambda\Iop),
	\label{GVon}
	\en where $\lambda\in\mathbb{R}\smallsetminus\{0\}$.
\end{theorem}
\begin{proof}
 Because $\lambda\in\mathbb{R}\smallsetminus\{0\}$ and the fact that $\D\left(\displaystyle\frac{1}{\lambda}\,A^{\dagger}\right)=\D(A^{\dagger})$, $\D\left(\displaystyle\frac{1}{\lambda}\,A\right)=\D(A)$ and $\ker\left(\displaystyle\frac{1}{\lambda}\,A^{\dagger}\pm\bi\Iop\right)=\ker(A^{\dagger}\pm\bi\lambda\Iop)$, (\ref{GVon}) follows by replacing  $A$ by $\left[ \displaystyle\frac{1}{\lambda}\,A\right] $ in Theorem \ref{Von}.
\end{proof}
\begin{remark}
In the statements of Lemma \ref{neq-i}, Propositions \ref{ken} and \ref{neq1-i}, and Theorems \ref{csadj_i}, \ref{Von} and \ref{G-Von}, if we replace $\bi$ by $\bj$ or $\bk$ the validation of these statements will not be affected.
\end{remark}
\begin{remark}
	Using (\ref{lft_mul-op}) and (\ref{sc_mul_aj-op}), for any $\bfrakq\in\quat$ we can obtain that $(\bfrakq \Iop)^{\dagger}=(\overline{\bfrakq} \Iop)$.
	For, let $\phi\in V_{\quat}^{R}$, then
\begin{eqnarray*}
(\bfrakq\Iop)^{\dagger}\phi&=&[(\Iop)^{\dagger}\overline{\bfrakq}]\phi
=\Iop(\overline{\bfrakq}\phi)=\overline{\bfrakq}\phi=\sum_{k\in N}\varphi_{k}\overline{\bfrakq}\langle \varphi_{k}\mid \phi\rangle\\
&=&\sum_{k\in N}\varphi_k\overline{\bfrakq}\langle\varphi_k\mid\Iop\phi\rangle=\overline{\bfrakq}(\Iop\phi) =(\overline{\bfrakq} \Iop)\phi.
\end{eqnarray*}
	That is, $(\bfrakq \Iop)^{\dagger}=(\overline{\bfrakq} \Iop)$. Similarly we can obtain that $(\Iop\bfrakq )^{\dagger}=(\Iop\overline{\bfrakq} )$.
\end{remark}
\begin{proposition}\label{Gen-Von-neq}
	Let $A:\D(A)\subseteq V_{\quat}^R\longrightarrow V_{\quat}^R$ be a densely defined right $\quat$-linear operator with the property that $\bi\phi,\bj\phi,\bk\phi\in \D(A)$, for all $\phi\in \D(A)$. If the operators $\bi A$, $\bj A$ and $\bk A$ are anti-symmetric, then
	\be
	\|(A-\bfrakq\Iop)\phi\|^{2}=\|(A-q_0\Iop)\phi\|^{2}+
	(q_1^{2} +  q_2^{2} + q_3^{2})\|\phi\|^{2},
	\label{GenVon-neq}
	\en for all $\phi\in\D(A)$, where  $\bfrakq = q_0 + \bi q_1 + \bj q_2 + \bk q_3\in\quat$.
\end{proposition}
\begin{proof}
	Let $\phi\in\D(A)$ and $\bfrakq = q_0 + \bi q_1 + \bj q_2 + \bk q_3\in\quat$.
	Since $\bi A$, $\bj A$ and $\bk A$ are anti-symmetric,  $\bi (A-q_0\Iop)$, $\bj (A-q_0\Iop)$ and $\bk (A-q_0\Iop)$ are anti-symmetric operators. For, using $(\bi A)^{\dagger}=-(\bi A)$ and $(\bi \Iop)^{\dagger}=-(\bi \Iop)$, we can easily see that $(\bi q_{0}\Iop)^{\dagger}=-\bi q_{0}\Iop$ as $q_{0}\in\mathbb{R}$, and
	$$[\bi (A-q_0\Iop)]^{\dagger}=[\bi A-\bi q_{0}\Iop]^{\dagger}=(\bi A)^{\dagger}-(\bi q_{0}\Iop)^{\dagger}=-(\bi A)+\bi q_{0}\Iop=-\bi (A-q_0\Iop).$$
	Similarly $[\bj (A-q_0\Iop)]^{\dagger}=-\bj (A-q_0\Iop)$ and $[\bk (A-q_0\Iop)]^{\dagger}=-\bk (A-q_0\Iop)$ on $\D(A)$.
	 Then for any $\phi\in\D(A)$, we have, using (\ref{sc_mul_aj-op}) and (d) in Proposition \ref{lft_mul},
\begin{eqnarray*}
\langle (A-q_0\Iop)\phi\mid (\bi\Iop)\phi \rangle
&=&\langle (A-q_0\Iop)\phi\mid \bi\phi \rangle=\langle \overline{\bi}((A-q_0\Iop)\phi)\mid \phi \rangle\\
&=&\langle -(\bi (A-q_0\Iop))\phi\mid \phi \rangle=\langle (\bi (A-q_0\Iop))^{\dagger}\phi\mid \phi \rangle\\
&=&\langle ((A-q_0\Iop)^{\dagger}\overline{\bi})\phi\mid \phi \rangle=\langle (A-q_0\Iop)^{\dagger}(\overline{\bi}\phi)\mid \phi \rangle\\
&=&\langle \overline{\bi}\phi\mid (A-q_0\Iop) \phi \rangle=-\langle \bi\phi\mid (A-q_0\Iop) \phi \rangle\\
&=&-\langle (\bi\Iop)\phi\mid (A-q_0\Iop)\phi \rangle.
\end{eqnarray*}
	Similarly, $\langle (A-q_0\Iop)\phi\mid \bj q_{1}\phi \rangle=-\langle \bj q_{1}\phi\mid (A-q_0\Iop) \phi \rangle$ and $\langle (A-q_0\Iop)\phi\mid \bk q_{1}\phi \rangle=-\langle \bk q_{1}\phi\mid (A-q_0\Iop) \phi \rangle$.
	Keeping these in mind, from a direct calculation, we can obtain
	\begin{eqnarray*}
		\|(A-\bfrakq\Iop)\phi\|^{2}
		&=&\langle(A-\bfrakq\Iop)\phi\mid(A-\bfrakq\Iop)\phi\rangle\\
		&=&\langle(A-q_0\Iop)\phi-(\bi q_1 + \bj q_2 + \bk q_3)\phi\mid(A-q_0\Iop)\phi-(\bi q_1 + \bj q_2 + \bk q_3)\phi\rangle\\
		&=&\|(A-q_0\Iop)\phi\|^{2}+
		(q_1^{2} +  q_2^{2} + q_3^{2})\|\phi\|^{2}.
	\end{eqnarray*}
	Hence the result follows.
\end{proof}
We can have the following Proposition as a direct consequence of the above Proposition.
\begin{proposition}\label{ken_Gen}
	Let $A:\D(A)\subseteq V_{\quat}^R\longrightarrow V_{\quat}^R$ be a densely defined right $\quat$-linear operator with the property that $\bi\phi,\bj\phi,\bk\phi\in \D(A)$, for all $\phi\in \D(A)$. If the operators $\bi A$, $\bj A$ and $\bk A$ are anti-symmetric and $\bfrakq = q_0 + \bi q_1 + \bj q_2 + \bk q_3 \in\quat$ with $q_1^{2} +  q_2^{2} + q_3^{2}\neq 0$, then  $$\ker{(A- \bfrakq\Iop)}=\{0\}~~\mbox{~and~}~\ker{(A- \overline{\bfrakq}\Iop)}=\{0\}.$$
\end{proposition}
\begin{proof}
Let $\phi\in\ker{(A- \bfrakq\Iop)}$. Then by (\ref{GenVon-neq}), we have
$$0=\|(A- \bfrakq\Iop)\phi\|^{2}=\|(A-q_0\Iop)\phi\|^{2}+(q_1^{2} +  q_2^{2} + q_3^{2})\|\phi\|^{2}.$$
This implies that $\phi=0$ as $q_1^{2} +  q_2^{2} + q_3^{2}\neq 0$. That is, $\ker{(A- \bfrakq\Iop)}=\{0\}$. Likewise, we can easily obtain that $\ker{(A- \overline{\bfrakq}\Iop)}=\{0\}$ by replacing $\bfrakq$ by $\overline{\bfrakq}$ in (\ref{GenVon-neq}).
\end{proof}
\begin{proposition}\label{qA}
Let $A:\D(A)\subseteq V_{\quat}^R\longrightarrow V_{\quat}^R$ be a densely defined right $\quat$-linear symmetric operator with the property that  $\bfrakq\phi\in \D(A)$, for all $\phi\in \D(A)$ and $\bfrakq \in\quat$.
\begin{itemize}
\item [(a)] If $\bi A$, $\bj A$ and $\bk A$ are anti-symmetric, then $(\bfrakq A)^{\dagger}=\overline{\bfrakq}A$ and $\bfrakq A=A\bfrakq$.
\item [(b)] If $\bfrakq A$ is anti-symmetric, then
$$\|(A- \overline{\bfrakq}\Iop)\phi\|^{2}=\|A\phi\|^{2}+|\bfrakq|^{2}\|\phi\|^{2}$$
for all $\phi\in\D(A)$.
\item [(c)] If $\overline{\bfrakq} A$ is anti-symmetric, then
$$\|(A-\bfrakq\Iop)\phi\|^{2}=\|A\phi\|^{2}+|\bfrakq|^{2}\|\phi\|^{2}$$
for all $\phi\in\D(A)$.
\end{itemize}
\end{proposition}
\begin{proof}
Let $\phi\in\D(A)$ and $\bfrakq=q_0 + \bi q_1 + \bj q_2 + \bk q_3 \in\quat$. First note that, since $\bi\phi,\bj\phi,\bk\phi\in \D(A)$, for all $\phi\in \D(A)$ and $\D(A)$ is a linear subspace of $V_{\quat}^R$, we have  $\bfrakq\phi,\overline{\bfrakq}\phi\in \D(A)$, for all $\phi\in \D(A)$.
Then using (\ref{sc_mul_aj-op}), we may have
\begin{eqnarray*}
(\bfrakq A)^{\dagger}\phi&=&
(A^{\dagger}\overline{\bfrakq})\phi=A^{\dagger}(\overline{\bfrakq}\phi)\\
&=&A^{\dagger}\left(\sum_{k\in N}\varphi_{k}\overline{\bfrakq}\langle \varphi_{k}\mid \phi\rangle\right)\\
&=&A^{\dagger}\left(\sum_{k\in N}\varphi_{k}(q_0 - \bi q_1 - \bj q_2 - \bk q_3)\langle \varphi_{k}\mid \phi\rangle\right)\\
&=&A^{\dagger}\left(\sum_{k\in N}\varphi_{k}\langle \varphi_{k}\mid \phi\rangle\right)q_0-A^{\dagger}\left(\sum_{k\in N}\varphi_{k}\bi \langle \varphi_{k}\mid \phi\rangle\right)q_1\\ &~&-A^{\dagger}\left(\sum_{k\in N}\varphi_{k}\bj \langle \varphi_{k}\mid \phi\rangle\right)q_2-A^{\dagger}\left(\sum_{k\in N}\varphi_{k}\bk \langle \varphi_{k}\mid \phi\rangle\right)q_3\\
&=&A^{\dagger}(\phi)q_0-A^{\dagger}(\bi\phi)q_1-A^{\dagger}(\bj\phi)q_2-A^{\dagger}(\bk\phi)q_3\\
&=&A^{\dagger}(\phi)q_0+[(A^{\dagger}\overline{\bi})\phi]q_1+[(A^{\dagger}\overline{\bj})\phi]q_2+[(A^{\dagger}\overline{\bk})\phi]q_3\\
&=&A^{\dagger}(\phi)q_0+[(\bi A)^{\dagger}\phi]q_1+[(\bj A)^{\dagger}\phi]q_2+[(\bk A)^{\dagger}\phi]q_3\\
&=&A(\phi)q_0-[(\bi A)\phi]q_1-[(\bj A)\phi]q_2-[(\bk A)\phi]q_3\\
&=&\left(\sum_{k\in N}\varphi_{k}\langle \varphi_{k}\mid A\phi\rangle\right)q_0-\left(\sum_{k\in N}\varphi_{k}\bi \langle \varphi_{k}\mid A\phi\rangle\right)q_1\\ &~&-\left(\sum_{k\in N}\varphi_{k}\bj \langle \varphi_{k}\mid A\phi\rangle\right)q_2-\left(\sum_{k\in N}\varphi_{k}\bk \langle \varphi_{k}\mid A\phi\rangle\right)q_3\\
&=&\left(\sum_{k\in N}\varphi_{k}(q_0 - \bi q_1 - \bj q_2 - \bk q_3)\langle \varphi_{k}\mid A\phi\rangle\right)\\
&=&\left(\sum_{k\in N}\varphi_{k}\overline{\bfrakq}\langle \varphi_{k}\mid A\phi\rangle\right)=(\overline{\bfrakq}A)\phi\\
\end{eqnarray*}
That is, $(\bfrakq A)^{\dagger}=\overline{\bfrakq}A$ on $\D(A)$. Thus, using (\ref{sc_mul_aj-op}), we can obtain that, on $\D(A)$, $\bfrakq A=(\overline{\bfrakq} A)^{\dagger}=A^{\dagger}\bfrakq=A\bfrakq$ as $A$ is symmetric. Hence the result (a) follows.
\begin{eqnarray*}
\|(A-\overline{\bfrakq}\Iop)\phi\|^{2}
&=&\langle(A-\overline{\bfrakq}\Iop)\phi\mid\langle(A-\overline{\bfrakq}\Iop)\phi\rangle\rangle\\
&=&\langle A\phi\mid A\phi \rangle-\langle A\phi\mid (\overline{\bfrakq}\Iop)\phi \rangle-\langle (\overline{\bfrakq}\Iop)\phi\mid A\phi \rangle+\langle (\overline{\bfrakq}\Iop)\phi\mid (\overline{\bfrakq}\Iop)\phi \rangle.
\end{eqnarray*}
For any $\phi\in\D(A)$, one can obtain that,
\begin{eqnarray*}
\langle A\phi\mid (\overline{\bfrakq}\Iop)\phi \rangle
&=&\langle A\phi\mid \overline{\bfrakq}\phi \rangle=\langle \bfrakq(A\phi)\mid \phi \rangle~\mbox{~by (d) in Proposition \ref{lft_mul}}\\
&=&\langle (\bfrakq A)\phi\mid \phi \rangle=\langle -(\bfrakq A)^{\dagger}\phi\mid \phi \rangle~\mbox{~as $\bfrakq A$ is anti-symmetric}\\
&=&-\langle (A^{\dagger}\overline{\bfrakq})\phi\mid \phi \rangle=-\langle A^{\dagger}(\overline{\bfrakq}\phi)\mid \phi \rangle~\mbox{~by \ref{sc_mul_aj-op}}\\
&=&-\langle \overline{\bfrakq}\phi\mid A \phi \rangle=-\langle (\overline{\bfrakq}\Iop)\phi\mid A\phi \rangle.
\end{eqnarray*}
Therefore, $$\|(A-\overline{\bfrakq}\Iop)\phi\|^{2}=\langle A\phi\mid A\phi \rangle+\langle (\overline{\bfrakq}\Iop)\phi\mid (\overline{\bfrakq}\Iop)\phi \rangle.$$
Thus result (b) follows. If we replace $\bfrakq$ by $\overline{\bfrakq}$ in the statement of (b), then we have (c). This completes the proof.
\end{proof}
The following Proposition generalizes the Theorem \ref{csadj_i}.
\begin{proposition}\label{csadj-gen}
Let $A:\D(A)\subseteq V_{\quat}^R\longrightarrow V_{\quat}^R$ be a densely defined symmetric right $\quat$-linear operator  with the property that $\bi\phi,\bj\phi,\bk\phi\in \D(A)$, for all $\phi\in \D(A)$. If the operators $\bi A$, $\bj A$ and $\bk A$ are anti-symmetric, then for any $\bfrakq = q_0 + \bi q_1 + \bj q_2 + \bk q_3\in\quat$ with $q_1^{2} +  q_2^{2} + q_3^{2}\neq 0$, the following statements are equivalent:
\begin{itemize}
\item[(a)] $A$ is self-adjoint.
\item[(b)] $A$ is closed and $\ker(A^{\dagger}-\bfrakq\Iop)=\{0\}$ and $\ker(A^{\dagger}-\overline{\bfrakq}\Iop)=\{0\}$.
\item[(c)] $\text{ran}(A-\bfrakq\Iop)=\text{ran}(A-\overline{\bfrakq}\Iop)=V_{\quat}^{R}$.
\end{itemize}
\end{proposition}
\begin{proof}
(a) $\Rightarrow$ (b): Since $A^{\dagger}=A$ is closed and Proposition \ref{ken_Gen}, (b) follows.\\
(b) $\Rightarrow$ (c): Now $\ker{(A^{\dagger}- \bfrakq\Iop)}=\{0\}$ implies with the Proposition \ref{cldop} that $\text{ran}(A- \overline{\bfrakq}\Iop)^{\bot}=\{0\}$. But from the equality (\ref{GenVon-neq})	$$\|(A-\overline{\bfrakq}\Iop)\phi\|=\sqrt{\|(A-q_0\Iop)\phi\|^{2}+
(q_1^{2} +  q_2^{2} + q_3^{2})\|\phi\|^{2}}\geq\sqrt{(q_1^{2} +  q_2^{2} + q_3^{2})}\|\phi\|,$$ this together with Proposition \ref{prrpt} implies that $\text{ran}(A-\overline{\bfrakq}\Iop)$ is closed. Thereby $\text{ran}(A-\overline{\bfrakq}\Iop)=V_{\quat}^{R}$. Likewise $\text{ran}(A-\bfrakq\Iop)=V_{\quat}^{R}$. \\
 (c) $\Rightarrow$ (a): Let $\xi\in\ker(A^{\dagger}-\bfrakq\Iop)$, then by (c) in Proposition \ref{cldop} we have $\xi\in\text{ran}(A-\overline{\bfrakq}\Iop)^{\bot}$. Using (c), one can directly say that $\xi=0$ and so $\ker(A^{\dagger}-\bfrakq\Iop)=\{0\}$. Similarly we can get that $\ker(A^{\dagger}-\overline{\bfrakq}\Iop)=\{0\}$. But, for any $\phi\in\D(A^{\dagger})$, there exists $\psi\in\D(A)$, such that $(A^{\dagger}-\bfrakq\Iop)\phi=(A-\bfrakq\Iop)\psi$ as $\text{ran}(A-\bfrakq\Iop)=V_{\quat}^{R}$. Thus we have $(A^{\dagger}-\bfrakq\Iop)(\phi-\psi)=0$, and so $\phi=\psi\in\D(A)$.This completes the proof.
\end{proof}
The following theorem provides a generalization to the decomposition (\ref{Von_eq}) .
\begin{theorem}\label{Gen-Von}
Let $A:\D(A)\subseteq V_{\quat}^R\longrightarrow V_{\quat}^R$ be a densely defined closed right $\quat$-linear operator with the property that $\bi\phi,\bj\phi,\bk\phi\in \D(A)$, for all $\phi\in \D(A)$. If $A$ is symmetric and $\bi A$, $\bj A$ and $\bk A$ are anti-symmetric, then
\be
\D(A^{\dagger})=\D(A)\oplus\ker(A^{\dagger}-\bfrakq\Iop)
\oplus\ker(A^{\dagger}-\overline{\bfrakq}\Iop),
\label{GenVon}
\en where $\bfrakq = q_0 + \bi q_1 + \bj q_2 + \bk q_3\in\quat$ with $q_1^{2} +  q_2^{2} + q_3^{2}\neq 0$.
\end{theorem}
\begin{proof}
Let $\D_{\bfrakq}=\ker(A^{\dagger}-\bfrakq\Iop)$ and $\D_{\overline{\bfrakq}}=\ker(A^{\dagger}-\overline{\bfrakq}\Iop)$ where $\bfrakq = q_0 + \bi q_1 + \bj q_2 + \bk q_3\in\quat$. Now the equality (\ref{GenVon-neq}) implies together with Proposition \ref{prrpt} that $\text{ran}(A-\bfrakq\Iop)$ and $\text{ran}(A-\overline{\bfrakq}\Iop)$ are closed and from (c) in the Proposition \ref{cldop}, we have  $\text{ran}(A-\overline{\bfrakq}\Iop)^{\bot}=\ker(A^{\dagger}-\bfrakq\Iop)$ and $\text{ran}(A-\bfrakq\Iop)^{\bot}=\ker(A^{\dagger}-\overline{\bfrakq}\Iop)$. To prove that $\D(A), \D_{\bfrakq}$ and $\D_{\overline{\bfrakq}}$ are linearly independent subspaces, take $(\phi_{0},\phi_{\bfrakq},\phi_{\overline{\bfrakq}})\in\D(A)\times\D_{\bfrakq}\times\D_{\overline{\bfrakq}}$ such that $\phi_{0}+\phi_{\bfrakq}+\phi_{\overline{\bfrakq}}=0$. Then, as $A$ is symmetric, we have
	\begin{eqnarray*}
		(A-\bfrakq\Iop)\phi_{0}&=&-[(A^{\dagger}-\bfrakq\Iop)\phi_{\bfrakq}
+(A^{\dagger}-{\bfrakq}\Iop)\phi_{\overline{\bfrakq}}]\\
		&=&-[(A^{\dagger}-\bfrakq\Iop)\phi_{\bfrakq}+(A^{\dagger}-\overline{\bfrakq}\Iop)\phi_{\overline{\bfrakq}}-2(\bi q_1 + \bj q_2 + \bk q_3)\phi_{\overline{\bfrakq}}]\\
		&=& 2(\bi q_1 + \bj q_2 + \bk q_3)\phi_{\overline{\bfrakq}}\mbox{~~as~~}
			(\phi_{\bfrakq},\phi_{\overline{\bfrakq}})\in\ker(A^{\dagger}
-\bfrakq\Iop)\times\ker(A^{\dagger}-\overline{\bfrakq}\Iop).
	\end{eqnarray*}
	Now $\phi_{\overline{\bfrakq}}\in\ker(A^{\dagger}-\overline{\bfrakq}\Iop)\subseteq\D(A^{\dagger})\subseteq V_{\quat}^R=\overline{\D(A)}$. That is, there exists a sequence $\{\phi_{n}\}$ in $\D(A)$ such that $\phi_{n}\longrightarrow\phi_{\overline{\bfrakq}}$ as $n\longrightarrow\infty$. From this, for each $k\in N$, $\varphi_{k}(\bi q_1 + \bj q_2 + \bk q_3)\langle \varphi_{k}\mid \phi_{n}\rangle\longrightarrow\varphi_{k}(\bi q_1 + \bj q_2 + \bk q_3)\langle \varphi_{k}\mid \phi_{\overline{\bfrakq}}\rangle$ as $n\longrightarrow\infty$ and so
	$$2(\bi q_1 + \bj q_2 + \bk q_3)\phi_n\longrightarrow2(\bi q_1 + \bj q_2 + \bk q_3)\phi_{\overline{\bfrakq}}$$
	as $n\longrightarrow\infty$. Now for each $n\in\mathbb{N}$, we have
	\begin{eqnarray*}
		(A^{\dagger}-\overline{\bfrakq}\Iop)(2\bfrakq\phi_n)
		&=& A^{\dagger}(2\bfrakq\phi_n)-\overline{\bfrakq}\Iop(2\bfrakq\phi_n)\\
		&=& 2[A^{\dagger}(\bfrakq\phi_n)-\overline{\bfrakq}(\bfrakq\phi_n)]\\
		&=& 2[A^{\dagger}(\bfrakq\phi_n)-\mid \bfrakq\mid^{2}\phi_n]\mbox{~~by (c) in Proposition \ref{lft_mul}}\\
		&=& 2[(A^{\dagger}\bfrakq)\phi_n-\mid \bfrakq\mid^{2}\phi_n]\mbox{~~by \ref{rgt_mul-op}}\\
		&=& 2[(\overline{\bfrakq} A)^{\dagger}\phi_n-\mid \bfrakq\mid^{2}\phi_n]\mbox{~~by \ref{sc_mul_aj-op}}\\
		%&=& 2[-(\bfrakq A)^{\dagger}\phi_n-\phi_n]\\
		&=& 2[(\bfrakq A)\phi_n-\mid \bfrakq\mid^{2}\phi_n] \mbox{~~by  (a) in Proposition \ref{qA}}\\
		&=& 2[\bfrakq (A\phi_n)-\bfrakq((\overline{\bfrakq}\Iop)\phi_n)]\\
		&=& 2\bfrakq[(A\phi_n)-(\overline{\bfrakq}\Iop)\phi_n]\mbox{~~by (a) in Proposition \ref{lft_mul}}\\
		&=& 2\bfrakq[(A-\overline{\bfrakq}\Iop)\phi_n].
	\end{eqnarray*}
	But, using the fact that $q_{0}\in\mathbb{R}$ and (e) in the Proposition \ref{lft_mul}, we have
	$(A^{\dagger}-\overline{\bfrakq}\Iop)(2q_0\phi_n)=2q_0[(A^{\dagger}-\overline{\bfrakq}\Iop)\phi_n]$ as  $(A^{\dagger}-\overline{\bfrakq}\Iop)$ is linear. Therefore
	$$(A^{\dagger}-\overline{\bfrakq}\Iop)(2(\bi q_1 + \bj q_2 + \bk q_3)\phi_n)=2(\bi q_1 + \bj q_2 + \bk q_3)[(A^{\dagger}-\overline{\bfrakq}\Iop)\phi_n].$$
	Now we need to show that $(A^{\dagger}-\overline{\bfrakq}\Iop)$ is a closed operator. For, let $\{\xi_{n}\}$ be a sequence such that $\xi_{n}\longrightarrow\xi$ as $n\longrightarrow\infty$ with $(A^{\dagger}-\overline{\bfrakq}\Iop)\xi_{n}\longrightarrow\eta$ as $n\longrightarrow\infty$. Then  $(\overline{\bfrakq}\Iop)\xi_{n}=\overline{\bfrakq}\xi_{n}\longrightarrow\overline{\bfrakq}\xi=(\overline{\bfrakq}\Iop)\xi$ as $n\longrightarrow\infty$ and so
	$$\|A^{\dagger}\xi_{n}-(\eta+\overline{\bfrakq}\xi)\|\leq\|(A^{\dagger}-\overline{\bfrakq}\Iop)\xi_{n}-\eta\|+\|(\overline{\bfrakq}\Iop)\xi_{n}-(\overline{\bfrakq}\Iop)\xi\|\longrightarrow 0$$
	as $n\longrightarrow\infty$. That is,$A^{\dagger}\xi_{n}\longrightarrow(\eta-\overline{\bfrakq}\xi)$  as $n\longrightarrow\infty$. But we know $A^{\dagger}$ is closed. Thus $A^{\dagger}\xi=\eta-\overline{\bfrakq}\xi$ or equivalently $(A^{\dagger}-\overline{\bfrakq}\Iop)\xi=\eta$. Therefore, by Proposition \ref{copcr}, $(A^{\dagger}-\overline{\bfrakq}\Iop)$ is closed.
	
	Assume that $(A^{\dagger}-\overline{\bfrakq}\Iop)\phi_n\longrightarrow\psi$ as $n\longrightarrow\infty$, for some $\psi\in V_{\quat}^R$. But $\phi_{n}\longrightarrow\phi_{\overline{\bfrakq}}$ as $n\longrightarrow\infty$. Thus $ \psi=(A^{\dagger}-\overline{\bfrakq}\Iop)\phi_{\overline{\bfrakq}}=0$ as $(A^{\dagger}-\overline{\bfrakq}\Iop)$ is a closed operator. That is, $(A^{\dagger}-\overline{\bfrakq}\Iop)\phi_n\longrightarrow 0$ as $n\longrightarrow\infty$. From this one can easily see that $2(\bi q_1 + \bj q_2 + \bk q_3)[(A^{\dagger}-\overline{\bfrakq}\Iop)\phi_n]\longrightarrow 0$ as $n\longrightarrow\infty$. Therefore we have $(A^{\dagger}-\overline{\bfrakq}\Iop)(2(\bi q_1 + \bj q_2 + \bk q_3)\phi_n)\longrightarrow 0$ as $n\longrightarrow\infty$ and $2(\bi q_1 + \bj q_2 + \bk q_3)\phi_n \longrightarrow 2(\bi q_1 + \bj q_2 + \bk q_3)\phi_{\overline{\bfrakq}} $ as $n\longrightarrow\infty$. Thus $(A^{\dagger}-\overline{\bfrakq}\Iop)(2(\bi q_1 + \bj q_2 + \bk q_3)\phi_{\overline{\bfrakq}})=0$ as $(A^{\dagger}-\overline{\bfrakq}\Iop)$ is closed. That is,
	$$2(\bi q_1 + \bj q_2 + \bk q_3)\phi_{\overline{\bfrakq}}\in\ker(A^{\dagger}-\overline{\bfrakq}\Iop)=\text{ran}(A-\bfrakq\Iop)^{\bot}.$$
Thereby $2(\bi q_1 + \bj q_2 + \bk q_3)\phi_{\overline{\bfrakq}}=0$ as  $\text{ran}(A-\bfrakq\Iop)\cap\text{ran}(A-\bfrakq\Iop)^{\bot}=\{0\}$. Since $q_1^{2} +  q_2^{2} + q_3^{2}\neq 0$, we have $\langle\varphi_{j}\mid\phi_{\overline{\bfrakq}}\rangle=0$, for all $j\in N$. Thus $\phi_{\overline{\bfrakq}}=0$. In a similar fashion we can prove $\phi_{\bfrakq}=0$, and so $\phi_{0}=0$. It is  trivial that $$\D(A)\oplus\ker(A^{\dagger}-\bfrakq\Iop)\oplus\ker(A^{\dagger}-\overline{\bfrakq}\Iop)\subseteq\D(A^{\dagger})$$
as $A$ is symmetric. Conversely, let $\phi\in\D(A^{\dagger})$, then since $$V_{\quat}^{R}=\text{ran}(A-\bfrakq\Iop)\oplus\text{ran}(A-\bfrakq\Iop)^{\bot}=\text{ran}(A-\bfrakq\Iop)\oplus\D_{\overline{\bfrakq}},$$
there exists $(\theta,\omega)\in\text{ran}(A-\bfrakq\Iop)\times\ker(A^{\dagger}-\overline{\bfrakq}\Iop)$ such that $(A^{\dagger}-\bfrakq\Iop)\phi=\theta+\omega$. Now there exists $\phi_{0}\in\D(A)$ such that $\theta=(A-\bfrakq\Iop)\phi_{0}$ and choose $\displaystyle\phi_{\overline{\bfrakq}}=\frac{-\text{Im}(\bfrakq)}{2\mid\text{Im}(\bfrakq)\mid^{2}}\,\omega$. Then $\phi_{\overline{\bfrakq}}\in\D_{\overline{\bfrakq}} =\ker(A^{\dagger}-\overline{\bfrakq}\Iop)$ and $$\displaystyle(\bi q_1 + \bj q_2 + \bk q_3)\phi_{\overline{\bfrakq}}=\text{Im}(\bfrakq)\phi_{\overline{\bfrakq}} =\text{Im}(\bfrakq)\left(\frac{-\text{Im}(\bfrakq)}{2\mid\text{Im}(\bfrakq)\mid^{2}}\,\omega\right)=-\frac{1}{2}\omega.$$
That is, we have there exists $(\phi_{0},-2(\bi q_1 + \bj q_2 + \bk q_3)\phi_{\overline{\bfrakq}})\in\D(A)\times\D_{\overline{\bfrakq}}$ such that
$$(A^{\dagger}-\bfrakq\Iop)\phi=(A-\bfrakq\Iop)\phi_{0}-2(\bi q_1 + \bj q_2 + \bk q_3)\phi_{\overline{\bfrakq}}.$$
That is, $$(A^{\dagger}-\bfrakq\Iop)(\phi-\phi_{0}-\phi_{\overline{\bfrakq}})=0.$$
Thus if we take $\phi_{\bfrakq}=\phi-\phi_{0}-\phi_{\overline{\bfrakq}}\in\D_{\bfrakq}$, then we have completed the proof.
\end{proof}
Let us present an example that satisfies the assumptions of the above theorem.
\begin{example}
Consider the $L^2$-spaces: the complex Hilbert space $\HI_\C=L^2_\C(\C, e^{-|z|^2}\frac{dz\wedge d\oz}{2\pi i})$ and the right quaternionic Hilbert space
 $\HI_\quat=L^2_\quat(\quat, e^{-|z|^2}\frac{dz\wedge d\oz}{2\pi i}d\omega)$, where $d\omega$ is a normalized Harr measure on $SU(2)$. For $z\in\C$, let
$$\phi_n(z)=\frac{z^n}{\sqrt{n!}}\quad\text{and}\quad \phi_n(\oz)=\frac{\oz^n}{\sqrt{n!}}.$$ Consider the operators
$$N=z\frac{\partial}{\partial z},\quad \overline{N}=\oz\frac{\partial}{\partial \oz}$$ defined respectively on
$$\D(N)=\text{complex-span}\{\phi_n(z)|n\in\N\}\quad\text{and}\quad \D(\overline{N})=\text{complex-span}\{\phi_n(\oz)|n\in\N\}.$$
$\D(N)$ and $\D(\overline{N})$ are subspaces of $\HI_\C$ and the operators act on the vectors as
$$N\phi_n(z)=n\phi_n(z)\quad \text{and}\quad \overline{N}\phi_n(\oz)=n\phi_n(\oz).$$
Let
$$\HI_{hol}=\overline{\text{complex-span}}~\D(N)\quad \text{and}\quad \HI_{a-hol}=\overline{\text{complex-span}}~\D(\overline N).$$
It is now clear that
$$N:\D(N)\subseteq\HI_{hol}\longrightarrow\HI_{hol}\quad\text{and}\quad
\overline{N}:\D(\overline{N})\subseteq\HI_{a-hol}\longrightarrow\HI_{a-hol}$$
are densely defined linear symmetric operators on the well-known Bargmann (analytic and anti-analytic) complex spaces $\HI_{hol}$ and $\HI_{a-hol}$ respectively and also these are subspaces of $\HI_\C$. Let
\begin{eqnarray*}
P&=&\di(N,\overline{N})\\
\Phi_n(z,\oz)&=&\di(\phi_n(z),\phi_n(\oz))\in\quat\\
\D(P)&=&\quat\text{-right-linear-span}\{\Phi_n(z,\oz)|n\in\N\}\subset\HI_\quat\\
\HI_{reg}&=&\overline{\quat\text{-right-linear-span}}~\D(P).
\end{eqnarray*}
The set $\D(P)$ is a subspace of $\HI_{\quat}$ and $\HI_{reg}$ is a right quaternionic subspace of $\HI_{\quat}$ consists of right regular functions. The operator $P:\D(P)\subseteq\HI_{reg}\longrightarrow\HI_{reg}$ is a densely defined right quaternionic symmetric operator. Clearly the following hold true:
\begin{enumerate}
\item[(a)] $i\phi, j\phi, k\phi\in \D(P)$ for all $\phi\in \D(P)$ in view of considering $i, j, k$ as complex imaginary units (or identifying $j, k$ with $i$).
\item[(b)] $iP, jP, kP$ are anti-symmetric operators.
\end{enumerate}
\end{example}
%%%%%%%%%%%%%%%%%%%%%%%%%%%%%%%%%%%%%%%%%%%%%%%%%%%%%%%%%%%%%%
\subsection{Deficiency Indices}
\begin{definition}\label{def_ind_i}
Let $\bfrake\in\{\bi,\bj,\bk\}$ and $A:\D(A)\longrightarrow V_{\quat}^R$ be a densely defined closed symmetric right $\quat$-linear operator in a Hilbert space $V_{\quat}^R$  with the property that $\bfrake\phi\in\D(A)$, for all $\phi\in \D(A)$ and $\bfrake A$ is anti-symmetric. Define $$n_{\pm}^{\bfrake}(A)=\dim\ker(A^{\dagger}\mp\bfrake\Iop).$$
The pairs $(n_{+}^{\bfrake}(A),n_{-}^{\bfrake}(A))$ is called  $\bfrake$-deficiency indices of $A$. Here it is admissible for the indices to take the value $\infty$.
\end{definition}
\begin{corollary}\label{se-ad-1}
Let a right $\quat$-linear operator $A:\D(A)\longrightarrow V_{\quat}^R$ be a densely defined closed symmetric operator with the property that $\bi\phi\in\D(A)$, for all $\phi\in \D(A)$ and $\bi A$ is anti-symmetric. $A$ is self-adjoint if and only if the deficiency indices $(n_{+}^{\bi}(A),n_{-}^{\bi}(A))=(0,0)$.
\end{corollary}
\begin{proof}
$A$ is self-adjoint if and only if $D(A)=\D(A^{\dagger})$. By (\ref{Von}) we have $\ker(A^{\dagger}\pm\bi\Iop)=\{0\}$ if and only if $D(A)=\D(A^{\dagger})$. But  $\ker(A^{\dagger}\pm\bi\Iop)=\{0\}$ if and only if $(n_{+}^{\bi}(A),n_{-}^{\bi}(A))=(0,0)$. Therefore $A$ is self-adjoint if and only if $(n_{+}^{\bi}(A),n_{-}^{\bi}(A))=(0,0)$.
\end{proof}
\begin{lemma}\label{dim}
If $E$ and $F$ are closed linear subspace of $V_{\quat}^R$ such that $\dim E<\dim F$, then there exists $\psi\in F\cap E^{\bot}$ with $\psi\neq0$.
\end{lemma}
\begin{proof}
The proof follows as a counterpart to the complex case (for the proof of complex case, see \cite{Schm}, pg 27).
\end{proof}
The following theorem characterizes the stability of the deficiency indices.
\begin{theorem}\label{def_con}
Let $A:\D(A)\subseteq V_{\quat}^R\longrightarrow V_{\quat}^R$ be a densely defined right $\quat$-linear operator with the property that $\bi\phi,\bj\phi,\bk\phi\in \D(A)$, for all $\phi\in \D(A)$. If the operators $\bi A$, $\bj A$ and $\bk A$ are anti-symmetric, then for any $\bfrakq= q_0 + \bi q_1 + \bj q_2 + \bk q_3\in\quat$ with $q_1^{2} +  q_2^{2} + q_3^{2}\neq 0$, we have
\be
\dim\ker(A^{\dagger}-\bfrakq\Iop)=n_{+}^{\bi}(A)~~\mbox{~~and~~}~~\dim\ker(A^{\dagger}-\overline{\bfrakq}\Iop)=n_{-}^{\bi}(A).
\en
\end{theorem}
\begin{proof}
Recall that $\text{ran}(A-\bfrakq\Iop)$ and $\text{ran}(A-\overline{\bfrakq}\Iop)$ are closed and then $$\text{ran}(A-\overline{\bfrakq}\Iop)^{\bot}=\ker(A^{\dagger}-\bfrakq\Iop)~~\mbox{~~and~~}~~ \text{ran}(A-\bfrakq\Iop)^{\bot}=\ker(A^{\dagger}-\overline{\bfrakq}\Iop).$$
Let $\bfrakq_{0}\in B(\bfrakq,\delta)$, where  $B(\bfrakq,\delta)$ is an open ball centered at $\bfrakq$ with radius $\delta=\sqrt{q_1^{2} +  q_2^{2} + q_3^{2}}$. Assume that $$\dim\ker(A^{\dagger}-\bfrakq\Iop)<\dim\ker(A^{\dagger}-\bfrakq_{0}\Iop).$$ Hence by the Lemma \ref{dim}, we have
there exists $\psi\in \ker(A^{\dagger}-\bfrakq_{0}\Iop)\cap \ker(A^{\dagger}-\bfrakq\Iop)^{\bot}$ with $\psi\neq0$. That is, $\psi\in\text{ran}(A-\bfrakq\Iop)$ and so $\psi=(A-\overline{\bfrakq}\Iop)\phi$ for some $\phi\in\D(A)$. Now we have
\be\label{ore}
\langle(A-\overline{\bfrakq}\Iop)\phi\mid(A-\overline{\bfrakq}_{0}\Iop)\phi\rangle=0.
\en
as $\psi\in\ker(A^{\dagger}-\bfrakq_{0}\Iop)$. The same equation (\ref{ore}) can be obtained even when  the inequality $\dim\ker(A^{\dagger}-\bfrakq_{0}\Iop)<\dim\ker(A^{\dagger}-\bfrakq\Iop)$ is considered. Now without loss of generality assume that $(A-\overline{\bfrakq}\Iop)\phi\neq0$ and using (\ref{ore}), we can derive, using (b) in the Proposition \ref{lft_mul},
\begin{eqnarray*}
\|(A-\overline{\bfrakq}\Iop)\phi\|^{2}
&=&\langle(A-\overline{\bfrakq}\Iop)\phi\mid(A-\overline{\bfrakq}_{0}\Iop)\phi+(\overline{\bfrakq}_{0}-\overline{\bfrakq})\phi\rangle\\
&\leq&\mid\bfrakq-\bfrakq_{0}\mid\|\phi\|\|(A-\overline{\bfrakq}\Iop)\phi\|.
\end{eqnarray*}
Thus $\|(A-\overline{\bfrakq}\Iop)\phi\|\leq\mid\bfrakq-\bfrakq_{0}\mid\|\phi\|$. Since $\phi\neq0$ and $\bfrakq_{0}\in B(\bfrakq,\delta)$ with the equality (\ref{GenVon-neq}), we have
$$\mid\bfrakq-\bfrakq_{0}\mid\|\phi\|<\delta\|\phi\|\leq\|(A-\overline{\bfrakq}\Iop)\phi\|\leq\mid\bfrakq-\bfrakq_{0}\mid\|\phi\|,$$ which is a contradiction. Therefore, for any $\bfrakq_{0}\in B(\bfrakq,\delta)$,
$$\dim\ker(A^{\dagger}-\bfrakq_{0}\Iop)=\dim\ker(A^{\dagger}-\bfrakq\Iop)=n_{\bfrakq}(A)~\mbox{~~(say)}.$$ Finally let $\bfrakx,\bfraky\in\quat\smallsetminus\mathbb{R}$ and $\mathscr{P}(\bfrakx,\bfraky)$ be a polygonal path from $\bfrakx$ to $\bfraky$. Now $\{B(\bfrakq,\delta)\,:\,\bfrakq\in\mathscr{P}(\bfrakx,\bfraky)\}$ forms a open cover of the compact set $\mathscr{P}(\bfrakx,\bfraky)$, so there exists a finite sub-cover $\{B(\bfrakq_{\tau},\delta)\,:\,\tau=1,2,\cdots,s\}$. Since $n_{\bfrakq}(A)$ is constant on each open ball $B(\bfrakq_{\tau},\delta)$, we conclude that $n_{\bfrakx}(A)=n_{\bfraky}(A)$. Hence the theorem follows.
\end{proof}
The Theorem \ref{def_con} enables us to have unique deficiency indices which are independent of imaginary units $\bi,\bj$ and $\bk$ for a given right $\quat$-linear closed symmetric operator $A:\D(A)\longrightarrow V_{\quat}^R$. That is, $(n_{+}^{\bfrake}(A),n_{-}^{\bfrake}(A))=$ constant ordered pair, for all $\bfrake\in\{\bi,\bj,\bk\}$. Further, we can use the phrase {\em deficiency indices of $A$} commonly instead of using {\em $\bfrake$-deficiency indices of $A$}, and simply denote it by $(n_{+}(A),n_{-}(A))$.
	\begin{proposition}\label{reso_R}
	Let $A:\D(A)\subseteq V_{\quat}^R\longrightarrow V_{\quat}^R$ be a densely defined right $\quat$-linear symmetric operator and $\bfrakq \in\quat$. If $\bi\phi,\bj\phi,\bk\phi\in \D(A)$, for all $\phi\in \D(A)$ and the operators $\bi A$, $\bj A$ and $\bk A$ are anti-symmetric, then
	$$R_{\bfrakq}(A)=A^{2}-2\text{Re}(\bfrakq)A+|\bfrakq|^{2}\Iop=(A-\bfrakq\Iop)(A-
\overline{\bfrakq}\Iop)=(A-\overline{\bfrakq}\Iop)(A-\bfrakq\Iop).$$
	\end{proposition}
	\begin{proof}
	Let $\phi\in\D(A^{2})$. Then
	\begin{eqnarray*}
	(A-\bfrakq\Iop)(A-\overline{\bfrakq}\Iop)\phi&=& A((A-\overline{\bfrakq}\Iop)\phi)-\bfrakq\Iop((A-\overline{\bfrakq}\Iop)\phi)\\
	&=& A^{2}\phi-A(\overline{\bfrakq}\phi)-\bfrakq(A\phi)+\bfrakq(\overline{\bfrakq}\phi)\\
	&=& A^{2}\phi-(A\overline{\bfrakq})\phi-(\bfrakq A)\phi+\mid\bfrakq\mid^{2}\phi\\
	&=& A^{2}\phi-(\overline{\bfrakq}A)\phi-(\bfrakq A)\phi+\mid\bfrakq\mid^{2}\phi\mbox{~~by (a) in Proposition \ref{qA}}\\
	&=& A^{2}\phi-(\overline{\bfrakq}+\bfrakq )A\phi+\mid\bfrakq\mid^{2}\phi\\
	&=& (A^{2}-2\text{Re}(\bfrakq)A+|\bfrakq|^{2}\Iop)\phi=R_{\bfrakq}(A)\phi.
	\end{eqnarray*}
	That is, $R_{\bfrakq}(A)=(A-\bfrakq\Iop)(A-\overline{\bfrakq}\Iop)$. Similarly $R_{\bfrakq}(A)=(A-\overline{\bfrakq}\Iop)(A-\bfrakq\Iop).$
	\end{proof}
\begin{theorem}\label{S-spec}
Let $A:\D(A)\subseteq V_{\quat}^R\longrightarrow V_{\quat}^R$ be a densely defined right $\quat$-linear closed symmetric operator with the property that $\bi\phi,\bj\phi,\bk\phi\in \D(A)$, for all $\phi\in \D(A)$. If the operators $\bi A$, $\bj A$ and $\bk A$ are anti-symmetric, then $A$ is self-adjoint if and only if the spherical spectrum $\sigma_S(A)\subset\mathbb{R}$.\\
Moreover, if $A$ is self-adjoint and $\bfrakq= q_0 + \bi q_1 + \bj q_2 + \bk q_3\in\quat\smallsetminus\mathbb{R}$, then $\bfrakq\in\rho_{S}(A)$ and $\|R_{\bfrakq}(A)^{-1}\|\leq(q_1^{2} +  q_2^{2} + q_3^{2})^{-1}$.
\end{theorem}
\begin{proof}The direct part follows from the Proposition \ref{Pr2}. Conversely suppose that $\sigma_S(A)\subset\mathbb{R}$. Let $\bfrakq = q_0 + \bi q_1 + \bj q_2 + \bk q_3\in\quat$ with $q_1^{2} +  q_2^{2} + q_3^{2}\neq 0$ and assume that $n_{+}(A)>0$.  By Theorem \ref{def_con}, we have $n_{+}(A)=n_{\bfrakq}(A)>0$. From the equality (\ref{GenVon-neq})	$$\|(A-\bfrakq\Iop)\phi\|=\sqrt{\|(A-q_0\Iop)\phi\|^{2}+
		(q_1^{2} +  q_2^{2} + q_3^{2})\|\phi\|^{2}}\geq(q_1^{2} +  q_2^{2} + q_3^{2})\|\phi\|,$$ this together with Proposition \ref{prrpt} implies that $\text{ran}(A-\overline{\bfrakq}\Iop)$ is closed and so $$\overline{\text{ran}(A-\overline{\bfrakq}\Iop)}=\ker(A^{\dagger}-\bfrakq\Iop)^{\bot}\neq V_{\quat}^R.$$ Further, using the equation (\ref{GenVon-neq}), we can obtain, for any $\phi\in \D(A^{2})$,
	\begin{eqnarray*}
		\|R_{\overline{\bfrakq}}(A)\phi\|^{2}=\|R_{\bfrakq}(A)\phi\|^{2}
		&=&\|(A-\bfrakq\Iop)(A-\overline{\bfrakq}\Iop)\phi\|^{2}\\
		&=&\|(A-q_0\Iop)(A-\overline{\bfrakq}\Iop)\phi\|^{2}+
		(q_1^{2} +  q_2^{2} + q_3^{2})\|(A-\overline{\bfrakq}\Iop)\phi\|^{2}\\
		&=&\|(A-q_0\Iop)(A-\overline{\bfrakq}\Iop)\phi\|^{2}\\
		&~&+(q_1^{2}+ q_2^{2} + q_3^{2})[\|(A-q_0\Iop)\phi\|^{2}+
		(q_1^{2} +  q_2^{2} + q_3^{2})\|\phi\|^{2}]\\
		&\geq& (q_1^{2} +  q_2^{2} + q_3^{2})^{2}\|\phi\|^{2}.
	\end{eqnarray*}
	That is, for any $\phi\in \D(A)$, $\|R_{\overline{\bfrakq}}(A)\phi\|\geq (q_1^{2} +  q_2^{2} + q_3^{2})\|\phi\|$, this suffices to say that $\ker(R_{\overline{\bfrakq}}(A))=\{0\}$. Thus $\overline{\bfrakq}\in\sigma_{rS}(A)$, that is $\sigma_S(A)\nsubseteq\mathbb{R}$ as $\overline{\bfrakq}\in\quat\smallsetminus\mathbb{R}$. This contradicts to the supposition. A similar contradiction could be obtained, if $n_{-}(A)>0$. Thus $(n_{+}(A),n_{-}(A))=(0,0)$. Hence $A$ is self-adjoint, by Corollary \ref{se-ad-1}.\\
Suppose that $A$ is self-adjoint and $\bfrakq= q_0 + \bi q_1 + \bj q_2 + \bk q_3\in\quat\smallsetminus\mathbb{R}$. From the above result, one can easily see that $\bfrakq\in\rho_{S}(A)$. Then $R_{\bfrakq}(A)^{-1}:\text{ran}(R_{\bfrakq}(A))\longrightarrow\D(A^{2})$ exists and bounded. Let $\psi\in\text{ran}(R_{\bfrakq}(A))$, then there exists $\phi\in\D(A^{2})$ such that $\psi=R_{\bfrakq}(A)\phi$.
But we have that
$$
\|R_{\bfrakq}(A)\phi\|^{2}
\geq(q_1^{2} +  q_2^{2} + q_3^{2})^{2}\|\phi\|^{2}.
$$
That is,  $\|R_{\bfrakq}(A)^{-1}\psi\|\leq (q_1^{2} +  q_2^{2} + q_3^{2})^{-1}\|\psi\|$, for all $\psi\in\D(R_{\bfrakq}(A)^{-1})=\text{ran}(R_{\bfrakq}(A))$. Thus,
$$\|R_{\bfrakq}(A)^{-1}\|=\sup_{\|\psi\|=1}\|R_{\bfrakq}(A)^{-1}\psi\|\leq \sup_{\|\psi\|=1} (q_1^{2} +  q_2^{2} + q_3^{2})^{-1}\|\psi\|=(q_1^{2} +  q_2^{2} + q_3^{2})^{-1}.$$
This completes the proof.
\end{proof}
\section{Invariance of Deficiency Indices}
In the previous chapter we defined the deficiency indices for a class of unbounded operators and also investigated some of its  properties related to S-spectrum using a left multiplication defined in terms of a particular basis of a right quaternionic Hilbert space. However, a natural question arises: whether the defined deficiency indices are invariant under the basis change. As an answer to this question, in this section, we show that the defined deficiency indices are invariant under the basis change.

Let $\mathfrak{O}=\{\vartheta_{k}\,\mid\,k\in N\}$ be a different Hilbert basis from the basis in (\ref{b1}) for $V_{\quat}^{R}$. For any given $\bfrakq\in\quat\smallsetminus\{0\}$ define the operators $\mathcal{L}_{\bfrakq}$ and $\mathfrak{L}_{\bfrakq}$  by $$\mathcal{L}_{\bfrakq}\phi=\bfrakq\cdot\phi=\sum_{k\in N}\varphi_{k}\bfrakq\langle \varphi_{k}\mid \phi\rangle,$$ and $$\mathfrak{L}_{\bfrakq}\phi=\bfrakq\ast\phi=\sum_{l\in N}\vartheta_{l}\bfrakq\langle \vartheta_{l}\mid \phi\rangle,$$ for all $\phi\in V_{\quat}^{R}$. We would like to remind to the reader that, up to the end of section 4 what we called $\bfrakq\phi$ is now written as $\bfrakq\cdot\phi$. We wrote it in this new way for individuating it from the other left-scalar-multiplication $\bfrakq\ast\phi$.

Define, for a fixed $\bfrakq\in\quat\smallsetminus\{0\}$, $\varDelta_{\bfrakq}:\text{ran}(\mathcal{L}_{\bfrakq})\longrightarrow\text{ran}(\mathfrak{L}_{\bfrakq})$ by
$$\varDelta_{\bfrakq}(\mathcal{L}_{\bfrakq}\phi)=\mathfrak{L}_{\bfrakq}\phi,$$
for all $\phi\in\text{ran}(\mathcal{L}_{\bfrakq})$.

Let $\mathcal{L}_{\bfrakq}\phi,\mathcal{L}_{\bfrakq}\psi\in\text{ran}(\mathcal{L}_{\bfrakq})$, where $\phi,\psi\in V_{\quat}^{R}$.
\begin{eqnarray*}
\mathcal{L}_{\bfrakq}\phi=\mathcal{L}_{\bfrakq}\psi
&\Leftrightarrow&\bfrakq\cdot(\phi-\psi)=0~~\mbox{~by (a) in Proposition \ref{lft_mul}}\\
&\Leftrightarrow&\mid\bfrakq\mid\|\phi-\psi\|=0~~\mbox{~by (b) in Proposition \ref{lft_mul}}\\
&\Leftrightarrow&\|\phi-\psi\|=0~~\mbox{~as}~~ \bfrakq\ne0\\
&\Leftrightarrow&\phi=\psi\\
&\Leftrightarrow&\bfrakq\ast\phi=\bfrakq\ast\psi\\
&\Leftrightarrow&\varDelta_{\bfrakq}(\mathcal{L}_{\bfrakq}\phi)=\mathfrak{L}_{\bfrakq}\phi=\mathfrak{L}_{\bfrakq}\psi=\varDelta_{\bfrakq}(\mathcal{L}_{\bfrakq}\psi).
\end{eqnarray*}
Thus $\varDelta_{\bfrakq}$ is well-defined and one to one.

Let $\bfrakx,\bfraky\in\quat$ and $\mathcal{L}_{\bfrakq}\phi,\mathcal{L}_{\bfrakq}\psi\in\text{ran}(\mathcal{L}_{\bfrakq})$. Then by (a) in Proposition \ref{lft_mul}
\begin{eqnarray*}
\varDelta_{\bfrakq}[(\mathcal{L}_{\bfrakq}\phi)\bfrakx+(\mathcal{L}_{\bfrakq}\psi)\bfraky]&=& \varDelta_{\bfrakq}[(\bfrakq\cdot\phi)\bfrakx+(\bfrakq\cdot\psi)\bfraky]\\
&=&\varDelta_{\bfrakq}[\bfrakq\cdot(\phi\bfrakx)+\bfrakq\cdot(\psi\bfraky)]\\
&=&\varDelta_{\bfrakq}[\mathcal{L}_{\bfrakq}(\bfrakq\cdot(\phi\bfrakx+\psi\bfraky))]\\
&=&\mathfrak{L}_{\bfrakq}[\bfrakq\cdot(\phi\bfrakx+\psi\bfraky)]\\
&=&\bfrakq\ast(\phi\bfrakx+\psi\bfraky)\\
&=&\bfrakq\ast(\phi\bfrakx)+\bfrakq\ast(\psi\bfraky)\\
&=&(\bfrakq\ast\phi)\bfrakx+(\bfrakq\ast\psi)\bfraky\\
&=&(\mathfrak{L}_{\bfrakq}\phi)\bfrakx+(\mathfrak{L}_{\bfrakq}\psi)\bfraky\\
&=&[\varDelta_{\bfrakq}(\mathcal{L}_{\bfrakq}\phi)]\bfrakx+[\varDelta_{\bfrakq}(\mathcal{L}_{\bfrakq}\psi)]\bfraky.
\end{eqnarray*}
Hence the right-linearity of $\varDelta_{\bfrakq}$ follows. Also one can trivially see that $\varDelta_{\bfrakq}$ is onto.

To prove the continuity of $\varDelta_{\bfrakq}$, let $\mathcal{L}_{\bfrakq}\varphi\in\text{ran}(\mathcal{L}_{\bfrakq})$ and $\varepsilon>0$. Now if $\mathcal{L}_{\bfrakq}\phi\in B(\mathcal{L}_{\bfrakq}\varphi,\delta)$, the open ball centered at $\mathcal{L}_{\bfrakq}\varphi$ with the radius $\delta$, and $\delta=\dfrac{\varepsilon}{\mid\bfrakq\mid}>0$, then
\begin{eqnarray*}
\|\varDelta_{\bfrakq}(\mathcal{L}_{\bfrakq}\phi)-\varDelta_{\bfrakq}(\mathcal{L}_{\bfrakq}\varphi)\|
&=& \|\mathfrak{L}_{\bfrakq}\phi-\mathfrak{L}_{\bfrakq}\varphi\|\\
&=&\|\bfrakq\ast\phi-\bfrakq\ast\varphi\|\\
&=& \mid\bfrakq\mid\|\phi-\varphi\|<\varepsilon,
\end{eqnarray*}
This concludes that $\varDelta_{\bfrakq}$ is continuous. Therefore $\varDelta_{\bfrakq}$ is an isomorphism. That is, $\text{ran}(\mathcal{L}_{\bfrakq})\approxeq\text{ran}(\mathfrak{L}_{\bfrakq})$.

Furthermore, $\varDelta_{\bfrakq}$ is an isometry. For, let $\phi,\psi\in V_{\quat}^{R}$, then we can obtain that
\begin{eqnarray*}
\langle\varDelta_{\bfrakq}(\mathcal{L}_{\bfrakq}\phi)\mid\varDelta_{\bfrakq}(\mathcal{L}_{\bfrakq}\psi)\rangle
&=&\langle\mathfrak{L}_{\bfrakq}\phi\mid\mathfrak{L}_{\bfrakq}\psi\rangle\\
&=&\langle\bfrakq\ast\phi\mid\bfrakq\ast\psi\rangle\\
&=&\langle\overline{\bfrakq}\ast(\bfrakq\ast\phi)\mid\psi\rangle~~\mbox{by (d) of Proposition \ref{lft_mul}}\\
&=&\mid\bfrakq\mid^{2}\langle\phi\mid\psi\rangle\\
&=&\langle\overline{\bfrakq}\cdot(\bfrakq\cdot\phi)\mid\psi\rangle\\
&=&\langle\bfrakq\cdot\phi\mid\bfrakq\cdot\psi\rangle~~\mbox{by (d) of Proposition \ref{lft_mul}}\\
&=&\langle\mathcal{L}_{\bfrakq}\phi\mid\mathcal{L}_{\bfrakq}\psi\rangle.
\end{eqnarray*}
Hence $\varDelta_{\bfrakq}$ is an isometric isomorphism. The following proposition provides some tools to prove the next Theorem.
\begin{proposition}\label{Pro_lft}
Let $A:\D(A)\subseteq V_{\quat}^R\longrightarrow V_{\quat}^R$ be a right $\quat$-linear operator. Then for any $\bfrakp,\bfrakq\in\quat$,
\begin{itemize}
\item [(a)] $\text{ran}(\mathcal{L}_{\bfrakq})=\text{ran}(\mathfrak{L}_{\bfrakq})=V_{\quat}^{R}$;
\item [(b)] $\bfrakp\ast(\bfrakq\cdot\phi)=(\bfrakp\bfrakq)\ast\phi$ and $\bfrakp\cdot(\bfrakq\ast\phi)=(\bfrakp\bfrakq)\cdot\phi$, for all $\phi\in V_{\quat}^{R}$;
\item [(c)] $\varDelta_{\bfrakq}(\text{ran}(A-\bfrakq\cdot\Iop))=\text{ran}(A-\bfrakq\ast\Iop)$;\vspace{-0.2cm}
\end{itemize}
where $(\bfrakq\cdot\Iop)\phi=\bfrakq\cdot\phi$ and $(\bfrakq\ast\Iop)\phi=\bfrakq\ast\phi$, for all $\phi\in V_{\quat}^{R}$.
\begin{proof}
Since, for any $\phi\in V_{\quat}^{R}$, by (c) in Proposition \ref{lft_mul}, we can write that $\phi=\bfrakq\cdot\psi$, where $\psi=\dfrac{\overline{\bfrakq}}{\mid\bfrakq\mid^{2}}\cdot\phi$ and $\bfrakq\in\quat\smallsetminus\{0\}$, therefore we have $\text{ran}(\mathcal{L}_{\bfrakq})=V_{\quat}^{R}$. Similarly we can obtain $\text{ran}(\mathfrak{L}_{\bfrakq})=V_{\quat}^{R}$. Hence (a) follows.\\
For any $\phi\in V_{\quat}^{R}$, we have
\begin{eqnarray*}
\bfrakp\ast(\bfrakq\cdot\phi)
&=&\sum_{l\in N}\vartheta_{l}\bfrakp\langle \vartheta_{l}\mid \bfrakq\cdot\phi\rangle\\
&=&\sum_{l\in N}\vartheta_{l}\bfrakp\langle \overline{\bfrakq}\cdot\vartheta_{l}\mid \phi\rangle~~\mbox{by (d) of Proposition \ref{lft_mul}}\\
&=&\sum_{l\in N}\vartheta_{l}\bfrakp\langle \vartheta_{l}\overline{\bfrakq}\mid \phi\rangle~~\mbox{by (f) of Proposition \ref{lft_mul}}\\
&=&\sum_{l\in N}\vartheta_{l}\bfrakp\bfrakq\langle \vartheta_{l}\mid \phi\rangle\\
&=&(\bfrakp\bfrakq)\ast\phi.
\end{eqnarray*}
That is, $\bfrakp\ast(\bfrakq\cdot\phi)=(\bfrakp\bfrakq)\ast\phi$. In a similar way we can obtain $\bfrakp\cdot(\bfrakq\ast\phi)=(\bfrakp\bfrakq)\cdot\phi$. Hence (b) follows.\\
Let   $\psi\in\varDelta_{\bfrakq}(\text{ran}(A-\bfrakq\cdot\Iop))$ and $\bfrakp=\dfrac{\overline{\bfrakq}}{\mid\bfrakq\mid^{2}}$, where $\bfrakq\ne0$, then there exists $\phi\in\D(A)$ such that
\begin{eqnarray*}
\psi&=&\varDelta_{\bfrakq}((A-\bfrakq\cdot\Iop)\phi),\\
&=&\varDelta_{\bfrakq}(\mathcal{L}_{\bfrakq}\xi);~~\mbox{~where~~} \xi=(\bfrakp\cdot A\phi-\phi)\\
&=&\mathfrak{L}_{\bfrakq}\xi\\
&=&\bfrakq\ast(\bfrakp\cdot A\phi-\phi)\\
&=&(A-\bfrakq\ast\Iop)\phi~~\mbox{by part (b)}.
\end{eqnarray*}
That is, $\varDelta_{\bfrakq}(\text{ran}(A-\bfrakq\cdot\Iop))\subseteq\text{ran}(A-\bfrakq\ast\Iop)$. On the other hand, take $\psi\in\text{ran}(A-\bfrakq\ast\Iop)\subseteq\text{ran}(\mathfrak{L}_{\bfrakq})$.
%Since $\varDelta_{\bfrakq}$ is onto, we have
Then there exists $\phi\in\D(A)$ such that $\psi=(A-\bfrakq\ast\Iop)\phi$. As we did earlier, it can be obtained that $\varDelta_{\bfrakq}((A-\bfrakq\cdot\Iop)\phi)=(A-\bfrakq\ast\Iop)\phi=\psi$. This concludes (c) and as well as the proof of this proposition.
\end{proof}
\end{proposition}
The following theorem concludes that the defined deficiency indices are invariant under change of basis.
\begin{theorem}
Let $A:\D(A)\subseteq V_{\quat}^R\longrightarrow V_{\quat}^R$ be a densely defined right $\quat$-linear operator with the property that $\bi\phi,\bj\phi,\bk\phi\in \D(A)$, for all $\phi\in \D(A)$. If the operators $\bi A$, $\bj A$ and $\bk A$ are anti-symmetric, then for any $\bfrakq= q_0 + \bi q_1 + \bj q_2 + \bk q_3\in\quat$ with $q_1^{2} +  q_2^{2} + q_3^{2}\neq 0$, we have
$$\dim\text{ran}(A-\bfrakq\cdot\Iop)^{\bot}=\dim\text{ran}(A-\bfrakq\ast\Iop)^{\bot}.$$
\end{theorem}
\begin{proof}
From (c) of Proposition \ref{Pro_lft}, we have $$\varDelta_{\bfrakq}(\text{ran}(A-\bfrakq\cdot\Iop))=\text{ran}(A-\bfrakq\ast\Iop).$$
This implies that $$[\varDelta_{\bfrakq}(\text{ran}(A-\bfrakq\cdot\Iop))]^{\bot}=[\text{ran}(A-\bfrakq\ast\Iop)]^{\bot}.$$
Now let us show that $$[\varDelta_{\bfrakq}(\text{ran}(A-\bfrakq\cdot\Iop))]^{\bot}=\varDelta_{\bfrakq}(\text{ran}(A-\bfrakq\cdot\Iop)^{\bot}).$$ Let $\psi\in[\varDelta_{\bfrakq}(\text{ran}(A-\bfrakq\cdot\Iop))]^{\bot}$, then we have $\langle\varDelta_{\bfrakq}\vartheta\mid\psi\rangle=0$, for all $\vartheta\in\text{ran}(A-\bfrakq\cdot\Iop)$. By (a) of Proposition \ref{Pro_lft} and from the fact that $\varDelta_{\bfrakq}$ is onto, there exists $\phi\in\text{ran}(\mathcal{L}_{\bfrakq})$ such that $\psi=\varDelta_{\bfrakq}\phi$. Thus $\langle\varDelta_{\bfrakq}\vartheta\mid\varDelta_{\bfrakq}\phi\rangle=0$, and since $\varDelta_{\bfrakq}$ is an isometry, we have $\langle\vartheta\mid\phi\rangle=\langle\varDelta_{\bfrakq}\vartheta\mid\varDelta_{\bfrakq}\phi\rangle=0$. That is, $\langle\vartheta\mid\phi\rangle=0$, for all $\vartheta\in\text{ran}(A-\bfrakq\cdot\Iop)$. This implies that $\phi\in\text{ran}(A-\bfrakq\cdot\Iop)^{\bot}$. This suffices to say that $\psi\in\varDelta_{\bfrakq}(\text{ran}(A-\bfrakq\cdot\Iop)^{\bot})$ as $\psi=\varDelta_{\bfrakq}\phi$. That is, $[\varDelta_{\bfrakq}(\text{ran}(A-\bfrakq\cdot\Iop))]^{\bot}\subseteq\varDelta_{\bfrakq}(\text{ran}(A-\bfrakq\cdot\Iop)^{\bot})$. Above argument works contrariwise too. Hence
$$\varDelta_{\bfrakq}(\text{ran}(A-\bfrakq\cdot\Iop)^{\bot})=[\text{ran}(A-\bfrakq\ast\Iop)]^{\bot}.$$ Since $\varDelta_{\bfrakq}$ is an isomorphism, the Theorem follows.
\end{proof}
%%%%%%%%%%%%%%%%%%%%%%%%%%%%%%%%%%%%%%%%%%%%

\end{document}